\documentclass[10pt, doublecolumn]{IEEEtran}
\usepackage{epsfig,latexsym}
\usepackage{float}
\usepackage{indentfirst}
\usepackage{amsmath}
\usepackage{bm}
\usepackage{amssymb}
\usepackage{times}
\usepackage{enumitem}
\usepackage{cite}
\usepackage{lastpage}
\usepackage{color} 
\usepackage{amsthm}
\usepackage{array}
\usepackage{algorithm}
\usepackage{algorithmic}
\usepackage{multirow}
\usepackage{empheq}

\newtheorem{lemma}{Lemma}

\begin{document}
\title{Pinching Antennas for Multiple Access in Multigroup Multicast Communications}
\author{Shan Shan, Chongjun Ouyang, Yong Li, and Yuanwei Liu
\thanks{Shan Shan and Yong Li are with the School of Information and Communication Engineering, Beijing University of Posts and Telecommunications, Beijing 100876, China (email: \{shan.shan, liyong\}@bupt.edu.cn). Chongjun Ouyang is with the School of Electronic Engineering and Computer Science, Queen Mary University of London, London E1 4NS, U.K. (email: c.ouyang@qmul.ac.uk). Yuanwei Liu is with the Department of Electrical and Electronic Engineering, The University of Hong Kong, Hong Kong (email: yuanwei@hku.hk).}
}
\IEEEaftertitletext{\vspace{-3em}}
 \maketitle
\begin{abstract} 
This paper aims to design multiple access (MA) schemes to improve the max-min fairness (MMF) for pinching antennas (PAs)-based multigroup multicast communications, where PA placement and resource allocation are jointly optimized. 
Specifically, three MA schemes are considered to facilitate the multicast transmission: i) treating interference as noise (TIN), ii) non-orthogonal multiple access (NOMA), and iii) time-division multiple access (TDMA) with two PA reconfiguration protocols, namely pinching switching (PS) and pinching multiplexing (PM). 
i) For TIN, a closed-form solution is derived for optimal power allocation, while a sequential element-wise optimization (SEO) is developed for the PA placement.
ii) For NOMA, a recursive power allocation framework incorporating a bisection search is developed, and a hierarchical objective evaluation (HOE) mechanism is incorporated to simplify the SEO process for PA location update.
iii) For TDMA, the PS protocol allows the PA locations to be optimized separately using the SEO method, after which the time-power allocation is solved as a convex problem with a global optimum. Under the PM protocol, the PA locations are jointly optimized with the time-power resources through a Karush-Kuhn-Tucker (KKT)-based analytical solution.
Numerical results demonstrate that: i) the pinching-antenna system (PASS) architecture significantly outperforms traditional fixed-antenna systems. ii) TDMA-PS achieves superior performance by fully leveraging the flexible PA reconfiguration and benefiting from interference-free transmission, whereas TIN serves as a practical lower-bound solution due to its simplicity despite its limited performance. iii) NOMA consistently outperforms TDMA-PM and, in high transmit power regimes with heterogeneous multicast group distributions, can even surpass the performance achieved by TDMA-PS.
\end{abstract}

\vspace{-10pt}
\section{Introduction}
Flexible-antenna systems have recently attracted extensive interest in wireless communications, given their superior capability to adaptively reconfigure the electromagnetic propagation environment~\cite{10910066, 10945421}. Representative architectures include fluid antennas~\cite{9264694, 9650760} and movable antennas~\cite{10286328, 10909572}, which dynamically adjust antenna positions to improve channel conditions. However, a fundamental limitation of these systems is that the antenna mobility is typically confined to a localized region spanning only a few wavelengths. While such local position refinement is effective in mitigating small-scale fading, it remains insufficient for compensating for severe large-scale path loss. Alternatively, reconfigurable intelligent surfaces (RISs) have been widely advocated to manipulate signal propagation via passive phase shifting~\cite{9424177 ,8741198}. While RISs are capable of establishing additional virtual line-of-sight (LoS) links for blocked users, they suffer from the severe double-fading effect. Specifically, the multiplicative path loss associated with the reflection link significantly restricts their energy efficiency over long distances. Consequently, the effective circumvention of large-scale path loss remains a persistent challenge for state-of-the-art flexible-antenna technologies.

In response to this challenge, the \emph{pinching-antenna system (PASS)} has recently emerged as a promising flexible-antenna paradigm~\cite{fukuda2022pinching, 10945421, 11169486}. PASS exploits low-loss dielectric waveguides for signal transmission, where dielectric perturbations, termed \emph{pinching antennas (PAs)}, can be flexibly deployed at arbitrary locations along an extended waveguide to couple electromagnetic energy into free space. Benefiting from the arbitrary long structure as well as the low attenuation characteristics of the waveguide, PASS can effectively bypass obstacles to establish strong LoS links with spatially dispersed users, thereby significantly mitigating large-scale path loss~\cite{11036558}. Furthermore, PASS exhibits superior scalability over conventional flexible-antenna architectures, as the addition or removal of PAs can be seamlessly implemented without altering the fundamental waveguide structure.
\vspace{-10pt}
\subsection{Related Works}
Motivated by these significant advantages, extensive research efforts have been dedicated to PASS. In particular, the foundational theoretical study in~\cite{10945421} investigated the fundamental transmission characteristics of PASS. Subsequently, \cite{10981775} analyzed the array gain of multi-PA scenarios and derived optimal deployment strategies regarding antenna number and inter-antenna spacing. From a signal modeling perspective, the authors in~\cite{wang2025modeling} introduced a physics-based directional coupler model for PA design, in which two wave-coupled power radiation models were derived based on electromagnetic theory.  Furthermore,~\cite{11263923} proposed an adjustable power radiation framework, which enables flexible radiation ratio control via coupling spacing adjustment to optimize discrete PA activation.

An inherent structural characteristic of PASS is that all PAs deployed along a single dielectric waveguide are fed by a common radio frequency (RF) chain, which results in the transmission of an identical signal stream. Consequently, extensive research efforts have been dedicated to investigating efficient multiple access (MA) techniques for multiuser communications using a single waveguide. 
For non-orthogonal multiple access (NOMA)-enabled PASS, the authors in~\cite{10912473, 11146800, 11204499, 11029492} focused on maximizing system throughput via joint PA placement and power allocation, whereas~\cite{11186151, 11131179} aimed to minimize the total power consumption.
Comparison between orthogonal multiple access (OMA) and NOMA has been conducted in~\cite{yanyu2025omanoma, qiaozi2025omanoma, 11016750, peng2025omanoma}. For instance, \cite{yanyu2025omanoma} analyzed the outage probability and diversity orders in a two-user scenario, while~\cite{qiaozi2025omanoma} and~\cite{peng2025omanoma} investigated power minimization and computation efficiency for PASS-enabled mobile edge computing (MEC) systems, respectively. 
These studies predominantly pursue performance metrics such as spectral efficiency or energy efficiency, which leaves user fairness largely unexplored. 
Towards this end, the authors in~\cite{11180028} introduced an orthogonal frequency-division multiple access (OFDMA)-based framework to achieve max-min fairness. In addition,~\cite{11114424} investigated fundamental physical-layer effects in PASS, including frequency-dependent waveguide attenuation and phase misalignment. 
Beyond conventional schemes, a PASS-specific environment-division multiple access (EDMA) technique was proposed in~\cite{zhiguo2025edma}, which optimizes PA locations to strategically obstruct interference paths via environmental blockage. 
Collectively, the aforementioned works validate the superiority of PASS in realizing flexible interference management and high-efficiency multiuser communications by leveraging its unique spatial reconfigurability.

\subsection{Motivations and Contributions}
Most existing studies on PASS have primarily focused on unicast transmission, where the single-feed waveguide architecture naturally supports the delivery of a single data stream. This structural characteristic also makes PASS inherently well-suited for physical-layer multicasting, since a common message can be efficiently conveyed through the shared electromagnetic waveguide to a group of spatially distributed users. This observation motivates the extension of PASS from unicast toward multicast transmissions, especially for emerging content-centric and group-oriented wireless services.
When multiple multicast groups are simultaneously served, however, the single-RF-chain architecture constrains the available spatial multiplexing degrees of freedom (DoFs), and inter-group interference becomes a dominant bottleneck. In such overloaded multigroup multicast scenarios, the design of suitable MA schemes plays a critical role in managing interference and ensuring reliable service delivery. Despite its practical importance, the behavior and performance of MA strategies under PASS-enabled multigroup multicast transmission remain insufficiently understood.

To address this knowledge gap, this work conducts a systematic investigation of three representative MA schemes for PASS-based multigroup multicast communications. The key contributions are summarized as follows:
\begin{itemize}
	\item We consider a PASS-enabled multigroup multicast system and develop a MA framework to manage inter-group interference. Within this framework, we explicitly characterize the minimum achievable group rate under three representative MA strategies, namely treating interference as noise (TIN), NOMA, and time-division multiple access (TDMA). For each strategy, we adopt a max-min fairness (MMF) objective and formulate a joint optimization problem in which the PA locations and the corresponding power and time allocation are optimized simultaneously.  

	\item We first analyze the widely recognized TIN scheme as a fundamental baseline, where inter-group interference is simply treated as additional noise at each receiver. Specifically, we derive a closed-form solution for the optimal power allocation and an analytical upper bound on the achievable multicast rate, thereby characterizing the performance limits imposed by inter-group interference in PASS. For PA placement, we show that under the MMF criterion, the optimal PA configuration is obtained by maximizing the harmonic mean of the channel gains, and this objective can be efficiently handled by a sequential element-wise optimization method.
	
	\item We then explore the NOMA scheme, where inter-group interference is managed by assigning a decoding order across groups and letting receivers decode and cancel the signals of some interfering groups via successive interference cancellation (SIC). For power allocation, we derive a closed-form optimal solution in the two-group scenario, while for the general multigroup case we exploit a recursive structure and embed the optimal power allocation into a bisection-based MMF rate search. For PA placement, SEO is employed in the two-group case, and in the multigroup case it is further combined with hierarchical objective evaluation (HOE) strategy to reduce the computational burden.

	\item We next investigate the TDMA scheme, where inter-group interference is eliminated by serving different multicast groups in orthogonal time slots. Two TDMA-based protocols are considered: PS and PM. In the PS protocol, the PA placement is decoupled from the power and time allocation. the PA locations of different groups are first optimized via sequential element-wise search to maximize their individual effective channel gains, and the remaining time–energy allocation is then solved as a convex optimization problem. In the PM protocol, a common PA configuration is shared across all groups. For the single-PA case, we derive closed-form expressions for the optimal power allocation and the PA location under equal time allocation. For multi-PA configurations, we adopt the HOE-based SEO algorithm of PA placements, where the group-specific power and time allocation is determined by a Karush-Kuhn-Tucker (KKT)-based resource allocation solver within each iteration.

	\item Finally, we provide numerical results to validate the effectiveness of the proposed schemes and to demonstrate the superiority of PASS over traditional fixed-array architectures in multigroup multicast scenarios. A comprehensive comparison across diverse system configurations further shows that: i) TDMA-PS is generally preferred because it eliminates inter-group interference and fully exploits group-wise spatial reconfigurability; ii) TIN serves as a fundamental performance lower bound while offering very low implementation complexity, which makes it attractive for systems with stringent hardware or processing constraints; and iii) NOMA consistently outperforms TDMA-PM and, in high-power regimes with spatially heterogeneous user distributions, emerges as a superior alternative that can even surpass the performance of TDMA-PS.
\end{itemize}

\subsection{Organization and Notations}
The remainder of this paper is organized as follows. Section~\ref{sec:System_Model} introduces the system model of PASS-enabled multigroup multicast communications and formulates the MMF problem. Sections~\ref{sec:TIN}, \ref{sec:NOMA}, and \ref{sec:TDMA} present the detailed designs and analyses for the TIN, NOMA, and TDMA schemes, respectively. Numerical results are presented in Section~\ref{sec:simulation}. Finally, Section~\ref{sec:conclusion} concludes this paper.

\subsubsection*{Notations}
Scalars, vectors, and matrices are represented by regular, bold lowercase, and  bold uppercase letters, respectively. The sets of complex numbers and real numbers are denoted by $\mathbb{C}$ and $\mathbb{R}$. The inverse and transpose operators are denoted by $(\cdot)^{-1}$ and $(\cdot)^{\sf T}$, respectively. For a vector ${\bf x}$, $[{\bf x}]_i$ denotes its $i$th element. ${\mathcal C}{\mathcal N}(a, b^2)$ denotes a circularly symmetric complex Gaussian distribution with mean $a$ and variance $b^2$. The statistical expectation operator is represented by ${\mathbb{E}}\{\cdot\}$. The absolute value and Euclidean norm are denoted by $|\cdot|$ and $\|\cdot\|$, respectively. The big-O notation is denoted by ${\mathcal O}(\cdot)$. 
\vspace{-10pt}
\section{System Model}\label{sec:System_Model}
\begin{figure}[!t]
\centering
\includegraphics[height=0.33\textwidth]{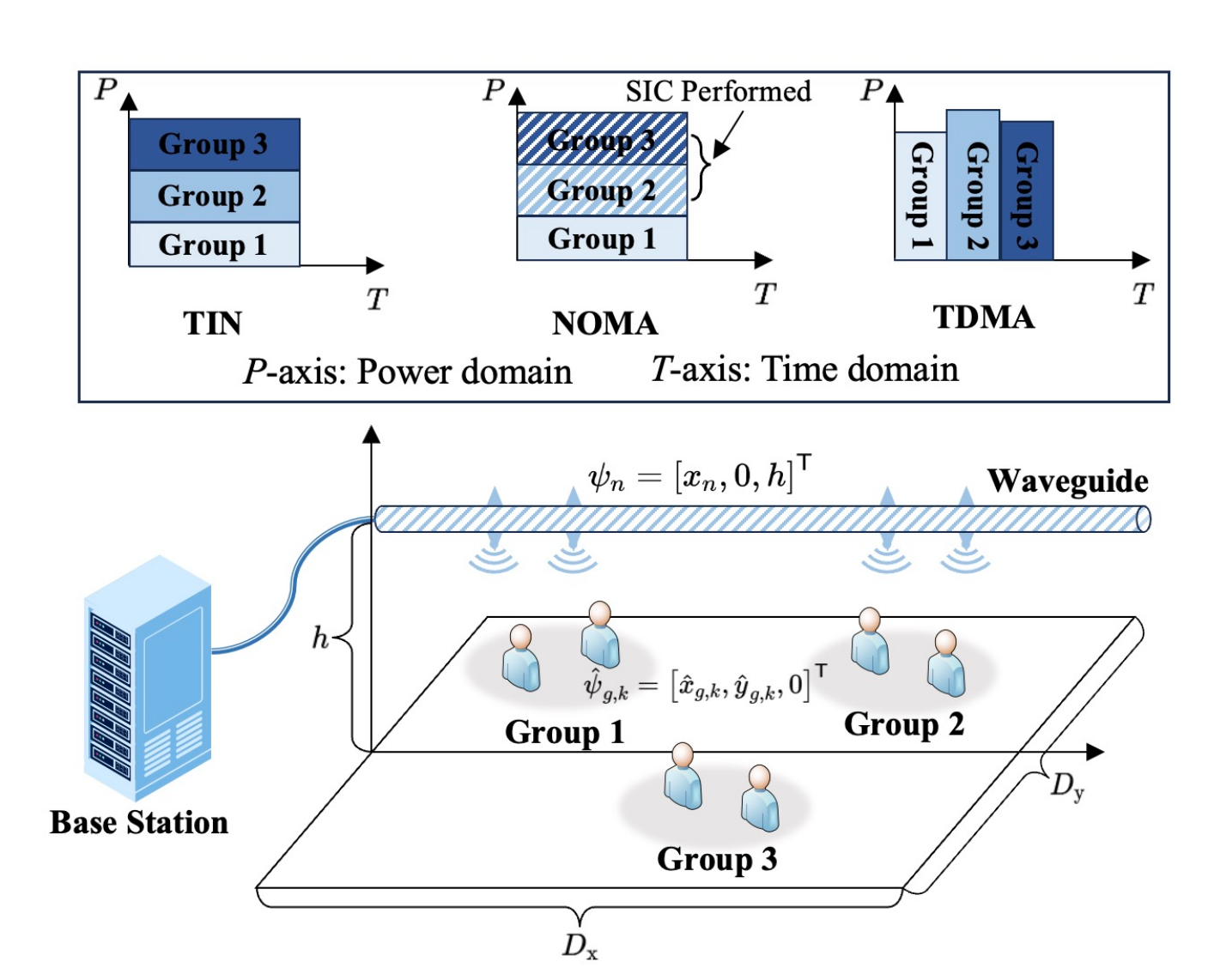}
\vspace{-10pt}
\caption{Illustration of the PASS-enabled multigroup multicast system and the associated multiple access strategies.}
\label{fig:System_Model}
\end{figure}
We consider a PASS-enabled multigroup multicast communication system. As illustrated in Fig. \ref{fig:System_Model}, the base station (BS) employs a single waveguide of length $D_{\rm x}$ to serve $G$ multicast groups. The waveguide is installed at a fixed height $h$, aligned parallel to the $x$-axis with $y_0$ being its $y$-coordinate. A single RF chain is connected to the waveguide, which up-converts the baseband signals to the carrier frequency and feeds them into the waveguide. The Cartesian coordinates of the feed point of the waveguide are denoted as ${\bm\psi}_{0} \!=\!\! [0, y_{0}, h]^{\sf T}$.
A total of $N$ PAs are activated along the waveguide, with the location of the $n$th PA denoted as ${\bm\psi}_{n} \!\!=\!\! [x_{n}, y_{0}, h]^{\sf T}$, where $x_n$ is its $x$-coordinate. Let ${\mathcal N}\!\!\triangleq\!\!\{1,\ldots,N\}$ denote the index set of the PAs. The vector of $x$-coordinates is defined as $\mathbf{x}\triangleq[x_{1}, \dots, x_{N}]^{\sf T} \in \mathbb{R}^{N \times 1}$, satisfying $0 \leqslant x_{1} < \dots < x_{N} \leqslant D_{\mathrm{x}}$. To mitigate electromagnetic mutual coupling, we enforce a minimum inter-PA spacing of $|x_{n} \!\!-\!\! x_{n-1}|\geq \Delta_{\sf min} = \lambda / 2$ for $n \!\!\in\!\! {\mathcal N}$ and $n\!\!>\!\!1$ \cite{10981775}, where $\lambda$ denotes the free-space wavelength.

In the multigroup multicast communication system, all users are distributed within a rectangular region of size $D = D_{\mathrm{x}} \times D_{\mathrm{y}}$. Users within the same group receive common desired message, while users across different groups are served with independent information. Let the set of multicast groups be denoted as $\mathcal{G}\triangleq\{1,2,\dots,G\}$. Assuming a total of $K$ users, the subset of users belonging to the group $g \in \mathcal{G}$ is denoted by $\mathcal{K}_{g}$, where each user is associated with exactly one group, i.e., $\mathcal{K}_{i} \cap \mathcal{K}_{j}=\emptyset, \forall i,j \in \mathcal{G},\ i \neq j$. The coordinate of the $k$th user in the $g$th group, denoted as the $(g,k)$th user, is given by ${\hat {\bm\psi}}_{g,k} = \left[{\hat x}_{g,k}, {\hat y}_{g,k}, 0\right]^{\sf T}$, $k\in{\mathcal K}_g.$
The in-waveguide signal propagation from the feed point ${\bm\psi}_0$ to the activated PAs is captured by the in-waveguide response vector
$\boldsymbol{\psi}(\mathbf{x})\in\mathbb{C}^{N\times 1}$, whose $n$th entry is given by
\begin{equation}\label{eq:psi_element}
[\boldsymbol{\psi}(\mathbf{x})]_n=\sqrt{\rho_n}\,{\rm e}^{-{\rm j}k_{\rm g}x_n},
\end{equation}
where $k_{\rm g}={2\pi}/{\lambda_{\rm g}}$ denotes the guided wavenumber,
$\lambda_{\rm g}=\lambda/n_{\rm eff}$ is the guided wavelength with $n_{\rm eff}$ being the effective refractive index~\cite{10945421},
and $\rho_n$ is the PA-dependent power-scaling factor.
For analytical simplicity, we assume equal power scaling across the $N$ PAs, i.e., $\rho_n=\rho=1/N$, $\forall n\in\mathcal{N}$~\cite{wang2025modeling}.

Given that PASS is a promising technology for high-frequency bands~\cite{fukuda2022pinching}, where LoS propagation typically dominates, we adopt a free-space LoS channel model.
Let $\mathbf{h}_{g,k}(\mathbf{x})\in\mathbb{C}^{N\times 1}$ denote the free-space channel vector from the $N$ PAs to user $(g,k)$, whose $n$th element is
\begin{equation}\label{eq:hk_element}
[\mathbf{h}_{g,k}(\mathbf{x})]_n
=\sqrt{\eta}\,D^{-1}_{g,k}(x_n)\,{\rm e}^{-{\rm j}k_0 D_{g,k}(x_n)},
\end{equation}
Here, $\eta=c^2/(16\pi^2 f_{\rm c}^2)$ with $c$ being the speed of light and $f_{\rm c}$ being the carrier frequency,
$k_0=2\pi/\lambda$ is the free-space wavenumber,
and $D_{g,k}(x_n)=\|{\hat{\bm\psi}}_{g,k}-{\bm\psi}_{n}\|$ represents the distance between the $n$th PA and user $(g,k)$.
Accordingly, the cascaded effective channel from the waveguide feed to user $(g,k)$ is defined as follows:
\begin{equation}\label{eq:effective_channel}
h_{g,k}(\mathbf{x}) \triangleq \mathbf{h}_{g,k}^{\sf T}(\mathbf{x})\,\boldsymbol{\psi}(\mathbf{x}) \in \mathbb{C}.
\end{equation}

Based on the proposed architecture, we investigate three representative MA schemes to support multigroup multicast transmissions: 1) TIN: All groups are served simultaneously on the same time-power resource, and each group decodes its intended stream by treating inter-group signals as noise. It offers a computationally efficient solution with low receiver complexity. 2) NOMA: Groups are superposed in the power domain and share the same time resources. The receivers apply successive interference cancellation (SIC) according to a decoding order, so the design jointly optimizes PA locations and power allocation to satisfy the signal-to-interference-plus-noise ratio (SINR) constraints induced by SIC. 3) TDMA: Groups are scheduled in orthogonal time slots with jointly optimized PA locations, per-slot power, and time allocation. Specifically, two protocols are considered, where PS allows slot-wise PA reconfiguration while PM uses a fixed PA placement for all slots.
In the following, we present the signal models for these three MA schemes.
\vspace{-10pt}
\subsection{Signal Transmission Model}\label{subsec:signal_model}
Let $s_g\in\mathbb{C}$ denote the independent Gaussian data symbol intended for multicast group $g$ with $\mathbb{E}\{|s_g|^2\}=1$, and let $P_g\in\mathbb{R}_+$ denote the transmit power allocated to the $g$th group. Define the power allocation vector as $\mathbf{p}\triangleq[P_1,\ldots,P_G]^{\sf T}$.

\subsubsection{TIN Scheme}\label{subsubsec:signal_tin}
Under TIN, the radiated signal is given by
\begin{equation}\label{eq:tx_TIN}
\mathbf{s}_{\sf TIN}=\boldsymbol{\psi}(\mathbf{x})\sum\nolimits_{g=1}^{G}\sqrt{P_g}\,s_g.
\end{equation}
The received signal at the $(g,k)$th user is
\begin{equation}\label{eq:rx_TIN}
y_{g,k}^{\sf TIN}=h_{g,k}(\mathbf{x})\sum\nolimits_{j=1}^{G}\sqrt{P_j}\,s_j+n_{g,k},
\end{equation}
where $n_{g,k}\sim\mathcal{CN}(0,\sigma_{g,k}^2)$ denotes the additive white Gaussian noise (AWGN) at the receiver. Accordingly, the SINR can be written as follows:
\begin{equation}\label{eq:SINR_TIN}
\gamma_{g,k}^{\sf TIN}(\mathbf{x},\mathbf{p})
=\frac{P_g|h_{g,k}(\mathbf{x})|^2}{\sum_{j\neq g}P_j|h_{g,k}(\mathbf{x})|^2+\sigma_{g,k}^2}.
\end{equation}

\subsubsection{NOMA Scheme}\label{subsubsec:signal_noma}
To implement NOMA, SIC is performed at the receivers, while all multicast groups are transmitted simultaneously via superposition coding as in \eqref{eq:tx_TIN}. For multigroup multicast, the decoding order is governed by the bottleneck channel condition of each group. Specifically, defining the user-wise channel-to-noise ratio (CNR) as $C_{g,k}(\mathbf{x}) \triangleq {|h_{g,k}(\mathbf{x})|^2}{\sigma_{g,k}^{-2}}$, the bottleneck group CNR is definded as follows:
\begin{equation}\label{eq:eff_channel_noma}
A_g(\mathbf{x}) \triangleq \min_{k\in\mathcal{K}_g} C_{g,k}(\mathbf{x}), \quad \forall g\in\mathcal{G}.
\end{equation}
Let $\bm{\pi}$ denote a decoding permutation that sorts the groups in ascending order of their bottleneck CNRs, i.e.,
\begin{equation}\label{eq:sorting}
A_{\pi(1)}(\mathbf{x}) \le A_{\pi(2)}(\mathbf{x}) \le \cdots \le A_{\pi(G)}(\mathbf{x}).
\end{equation}

According to the NOMA principle, after canceling the messages of weaker groups $\{\pi(1),\ldots,\pi(g-1)\}$, the SINR for user $k$ in group $\pi(g)$ to decode its intended message is
\begin{equation}\label{eq:SINR_NOMA}
\gamma_{\pi(g),k}^{\sf NOMA}(\mathbf{x},\mathbf{p})
=\frac{P_{\pi(g)} |h_{\pi(g),k}(\mathbf{x})|^2}
{\sum\nolimits_{j=g+1}^{G} P_{\pi(j)} |h_{\pi(g),k}(\mathbf{x})|^2 + \sigma_{\pi(g),k}^2 } .
\end{equation}
Since group $\pi(g)$ transmits a common codeword, its achievable rate is bottlenecked by the user with the minimum self-decoding SINR. Due to the monotonic relationship between the SINR in \eqref{eq:SINR_NOMA} and the CNR, this bottleneck user can be equivalently identified by
\begin{equation}\label{eq:bottleneck_user_noma}
k_g^\star \in \arg\min_{u\in\mathcal{K}_{\pi(g)}} C_{\pi(g),u}(\mathbf{x}).
\end{equation}

To successfully perform SIC, a user $(\pi(t),k)$ in a stronger group ($t \ge g$) must first decode the message of the weaker group $\pi(g)$ with the following SINR:
\begin{equation}\label{eq:SINR_cross}
\gamma^{\sf NOMA}_{\pi(t),k\rightarrow \pi(g)}(\mathbf{x},\mathbf{p})
=\frac{P_{\pi(g)}|h_{\pi(t),k}(\mathbf{x})|^2}
{\sum_{j=g+1}^{G}P_{\pi(j)}|h_{\pi(t),k}(\mathbf{x})|^2+\sigma_{\pi(t),k}^2}.
\end{equation}
The SIC feasibility requires that $\gamma^{\sf NOMA}_{\pi(t),k\rightarrow \pi(g)}(\mathbf{x},\mathbf{p}) \ge \gamma_{\pi(g),k_g^\star}^{\sf NOMA}(\mathbf{x},\mathbf{p})$ for all $t \ge g$. Notably, under the ascending CNR-based decoding order in \eqref{eq:sorting}, this feasibility condition is automatically satisfied. We retain its explicit formulation here solely for mathematical completeness.

\subsubsection{TDMA Scheme}\label{subsubsec:signal_tdma}
Under TDMA, groups are served over orthogonal time slots. Let $\tau_g\in(0,1]$ denote the time fraction assigned to group $g$, and define $\boldsymbol{\tau}\triangleq[\tau_1,\ldots,\tau_G]^{\sf T}$, which satisfies $\sum_{g=1}^{G}\tau_g\le 1$.
During the time slot of the $g$th group, the transmitted signal becomes
\begin{equation}\label{eq:tx_TDMA}
\mathbf{s}_{g}^{\sf TDMA}=\boldsymbol{\psi}(\mathbf{x})\sqrt{P_g}\,s_g,
\end{equation}
and the $(g,k)$th user receives
\begin{equation}\label{eq:rx_TDMA}
y_{g,k}^{\sf TDMA}=h_{g,k}(\mathbf{x})\sqrt{P_g}\,s_g+n_{g,k},
\end{equation}
which yields the signal-to-noise ratio (SNR) as follows\footnote{For the TDMA-PS scheme, the PA locations will be specialized to slot-dependent variables $\{{\bf x}_{g}\}$, whereas $\mathbf{x}$ here is used as a unified notation for the generic TDMA signal model.}:
\begin{equation}\label{eq:SNR_TDMA}
\gamma_{g,k}^{\sf TDMA}(\mathbf{x},\mathbf{p})
=\frac{P_g|h_{g,k}(\mathbf{x})|^2}{\sigma_{g,k}^2}.
\end{equation}

\vspace{-10pt}
\subsection{Max-Min Fairness Problem Formulation}\label{subsec:mmf_formulation}
Since each multicast group carries a common message, the achievable rate of group $g$ is determined by its worst user. 
To unify the MMF formulation across the considered MA schemes, we write the achievable group rate as follows:
\begin{equation}\label{eq:rate_general}
R_g(\mathbf{x},\mathbf{p},\boldsymbol{\tau})
=\tau_g\min_{k\in\mathcal{K}_g}\log_2\!\left(1+\gamma_{g,k}(\mathbf{x},\mathbf{p})\right),
\end{equation}
where $\gamma_{g,k}(\mathbf{x},\mathbf{p})$ is instantiated by \eqref{eq:SINR_TIN} for TIN, by \eqref{eq:SINR_NOMA} for NOMA, and by \eqref{eq:SNR_TDMA} for TDMA.
The time-allocation vector $\boldsymbol{\tau}$ serves as a generalized resource variable. Specifically, for TIN and NOMA, all groups occupy the full frame and thus $\tau_g=1$, $\forall g$. For TDMA, $\tau_g$ denotes the slot length with $\sum_{g=1}^{G}\tau_g\le 1$.

Our objective is to maximize the minimum achievable multicast rate among all groups. The joint design of PA locations and resource allocation is formulated as follows:
\begin{subequations}\label{prob:MMF_general}
\begin{align}
\mathcal{P}_0:\ 
&\max_{\mathbf{x}\in\mathcal{X}, \mathbf{p},\boldsymbol{\tau}}
\ \min_{g\in\mathcal{G}} \ R_g(\mathbf{x},\mathbf{p},\boldsymbol{\tau})\\
\mathrm{s.t.}\ 
&\sum\nolimits_{g=1}^{G} \tau_g P_g \le P_{\rm t}, \quad  P_g \ge 0,\ \forall g,\label{st:power} \\
&\sum\nolimits_{g=1}^{G} \tau_g \le 1, \quad 0<\tau_g\le 1,\ \forall g, \label{st:time} 
\end{align}
\end{subequations}
where $P_{\rm t}$ denotes the total transmit power budget normalized over one frame.
Moreover, $\mathcal{X}$ denotes the feasible set of PA locations specified by the waveguide aperture and the minimum inter-PA spacing constraint, i.e.,
\begin{align}\label{eq:X_feasible}
\mathcal{X} \!\!\triangleq\! \left\{ \!\mathbf{x} \!\in \!\mathbb{R}^N \Big| \
\begin{aligned}
& 0 \le x_n \le D_{\rm x},\ \forall n\in\mathcal{N},\\
& x_n - x_{n-1} \ge \Delta_{\sf min},\ \forall n\in\mathcal{N},\ n\ge 2
\end{aligned}
\right\}.
\end{align}

Problem~\eqref{prob:MMF_general} is a non-convex and non-smooth optimization problem, complicated by: i) the non-convexity of the objective function, which arises from the intricate coupling between resource variables ($\mathbf{p}, \boldsymbol{\tau}$) and the highly non-linear PA position-dependent channel gains $|h_{g,k}(\mathbf{x})|^2$; and ii) the non-smoothness induced by the nested multicast bottleneck $\min_{k\in\mathcal{K}_g}(\cdot)$ and group fairness $\min_{g\in\mathcal{G}}(\cdot)$ operators. 
To address these challenges, we develop joint PA locations and resource allocation design for TIN, NOMA, and TDMA schemes in the following sections to achieve the MMF objective.
\vspace{-10pt}
\section{TIN-Based Transmission}\label{sec:TIN}
We first study the TIN scheme, where each user decodes only its intended multicast signal while treating the inter-group interference as additional noise.
Moreover, for a given interference-plus-noise covariance, modeling the residual interference as \emph{Gaussian} is a conservative choice, since Gaussian noise minimizes the mutual information among all additive noises. Therefore, treating interference as Gaussian noise yields a worst-case achievable-rate lower bound~\cite{Hassibi2003HowMuchTraining}.
Accordingly, the received SINR of the $(g,k)$th user in \eqref{eq:SINR_TIN} can be equivalently rewritten as follows:
\begin{equation}\label{eq:SINR_TIN_equiv}
\gamma_{g,k}^{\sf TIN}(\mathbf{x},\mathbf{p})
={P_g}{\Big(\sum_{j\neq g}P_j+\sigma_{g,k}^2/\big|h_{g,k}(\mathbf{x})\big|^2\Big)^{-1}},
\end{equation}
which highlights that, under the single-waveguide architecture, the inter-group interference at a given user scales identically with $\big|h_{g,k}(\mathbf{x})\big|^2$.
Accordingly, problem~\eqref{prob:MMF_general} is specialized to the TIN scheme as follows:
\begin{subequations}\label{prob:MMF_TIN}
\begin{align}
\mathcal{P}^{\sf TIN}:
\max_{\mathbf{x}\in\mathcal{X},\,\mathbf{p}}~~
& \min_{g\in\mathcal{G}}~ R_g^{\sf TIN}(\mathbf{x},\mathbf{p}) \\
\text{s.t.}~~
& \sum\nolimits_{g=1}^{G} P_g \le P_{\rm t}, \ P_g\ge 0,\ \forall g\in\mathcal{G},
\end{align}\end{subequations}
where
\begin{align}
	R_g^{\sf TIN}(\mathbf{x},\mathbf{p})
= \min_{k\in\mathcal{K}_g}\log_2\!\left(1+\gamma_{g,k}^{\sf TIN}(\mathbf{x},\mathbf{p})\right).
\end{align}
Although the time-allocation vector ${\boldsymbol{\tau}}$ is absent, problem~\eqref{prob:MMF_TIN} remains non-convex due to the coupled interference structure in \eqref{eq:SINR_TIN_equiv} and the highly non-linear dependence of $|h_{g,k}(\mathbf{x})|^2$ on the PA locations. 
In what follows, we first obtain the optimal MMF power allocation in closed form for a given $\mathbf{x}$, and then optimize $\mathbf{x}$ via a sequential element-wise search.
\vspace{-10pt}
\subsection{Optimal Power Allocation Strategy}
Given a fixed PA configuration $\mathbf{x}$, the original problem reduces to a MMF power allocation problem. 
Since the multicast rate of the $g$th group is limited by its worst user, the MMF objective depends on each group only through the bottleneck CNR $A_g(\mathbf{x})$. 
By substituting \eqref{eq:eff_channel_noma} into \eqref{eq:SINR_TIN_equiv}, the original MMF power-allocation subproblem can be reduced to an equivalent MMF problem for $G$ \emph{virtual} users, where the effective SNR of the $g$th group is given by
\begin{equation}\label{eq:gamma_virtual_TIN}
\gamma_g(\mathbf{p})\!= \frac{P_g}{\sum_{j\neq g} P_j + A_g^{-1}(\mathbf{x})}
= \frac{P_g A_g(\mathbf{x})}{\big(P_{\rm t}-P_g\big)A_g(\mathbf{x}) + 1}.
\end{equation}
Accordingly, we derive the closed-form optimal power allocation in the following lemma.

\begin{lemma}\label{prop:closed_form_power_TIN}
For a fixed PA placement vector $\mathbf{x}$, the optimal MMF power allocation policy $\mathbf{p}^{\star}$ under the TIN scheme is given in a closed form. Specifically, the optimal equalized TIN SINR for each group can be written as follows:
\begin{equation}\label{eq:opt_gamma}
\gamma^{\sf TIN, \star}(\mathbf{x}) = \Big({\sum\nolimits_{g=1}^{G} \left(1 + {P^{-1}_{\rm t} A^{-1}_g(\mathbf{x})}\right) - 1}\Big)^{-1},
\end{equation}
and the optimal transmit power allocated to the $g$th group is
\begin{equation}\label{eq:opt_power}
P_g^{\sf TIN, \star} = \frac{\gamma^{\sf TIN, \star}(\mathbf{x})\left( P_{\rm t} + {A^{-1}_g(\mathbf{x})} \right)}{1+\gamma^{\sf TIN, \star}(\mathbf{x})}.
\end{equation}
\end{lemma}
\begin{IEEEproof}
Please refer to Appendix~\ref{app:proof_prop1} for more details.
\end{IEEEproof}

Lemma~\ref{prop:closed_form_power_TIN} reveals an intrinsic performance upper bound of the TIN-based transmission. 
Specifically, in the high-SNR regime where $P_{\rm t}$ is sufficiently large such that $P_{\rm t}A_g(\mathbf{x})\!\gg\!1$ for all $g\in\mathcal{G}$, the system becomes interference-limited rather than noise-limited. In this asymptotic regime, the noise terms in \eqref{eq:opt_gamma} become negligible, and the {optimal equalized TIN SINR saturates to $\gamma^{\sf TIN, \star}\approx \frac{1}{G-1}$. 
Consequently, the achievable multicast rate is theoretically bounded by
\begin{equation}\label{eq:tin_ceiling_rate}
\min_{g\in\mathcal{G}}~ R_g^{\sf TIN}(\mathbf{x},\mathbf{p}^{\star})
\le
\log_2\!\left(1+({G-1})^{-1}\right).
\end{equation}
This upper bound depends solely on the number of multicast groups $G$ and is independent of the PA locations $\mathbf{x}$. This implies that in the high-power regime, optimizing PA locations yields diminishing gains, as the performance is bottlenecked by the unavoidable inter-group interference inherent to the superposition transmission.
\vspace{-10pt}
\subsection{Sequential Element-Wise PA Location Optimization}\label{subsec:algorithm}
With the closed-form optimal power allocation $\mathbf{p}^{\star}$ derived in Lemma~\ref{prop:closed_form_power_TIN}, the original joint optimization problem $\mathcal{P}^{\sf TIN}$ can be reduced to a problem solely dependent on the PA locations $\mathbf{x}$.
Substituting $\gamma^{\sf TIN, \star}(\mathbf{x})$ from \eqref{eq:opt_gamma} into the objective of \eqref{prob:MMF_TIN}, maximizing the multicast rate is equivalent to minimizing the denominator term in \eqref{eq:opt_gamma}.
Consequently, the PA location optimization problem can be reformulated as follows:
\begin{align}\label{prob:x_obj}
\min_{\mathbf{x}\in\mathcal{X}} \,\, f_A(\mathbf{x}) \triangleq \sum\nolimits_{g=1}^{G} {A^{-1}_g(\mathbf{x})}.
\end{align}

Problem \eqref{prob:x_obj} involves optimizing the continuous-valued vector $\mathbf{x}$ over a highly non-convex landscape characterized by severe fluctuations arising from the superposition of multiple sinusoidal components. To circumvent the high computational complexity of continuous search, we discretize the waveguide aperture into a finite set of candidate locations. Specifically, we define a uniform grid $\mathcal{S}_x$ with $L$ sampling points as follows:
\begin{equation}
\mathcal{S}_x = \left\{ 0, \frac{D_{\rm x}}{L-1}, \frac{2 D_{\rm x}}{L-1}, \ldots, D_{\rm x} \right\}.
\end{equation}
Based on this discretization, we propose an element-wise sequential optimization strategy to iteratively refine the PA locations. In each iteration, we update the $n$th PA location $x_n$, while fixing the positions of the other $N-1$ PAs, denoted by $\mathbf{x}_{-n}$. Accordingly, the specific feasible set for updating $x_n$ is defined as follows:
\begin{align}\label{eq:feasible_set_new}
\mathcal{X}_n = \left\{ x \in \mathcal{S}_x \;\Big|\; |x - x_j| \ge \Delta_{\sf min}, \forall x_j \in \mathbf{x}_{-n} \right\}.
\end{align}
For a fixed $\mathbf{x}_{-n}$, the aggregate channel $h_{g,k}(\mathbf{x})$ can be efficiently decomposed into a fixed component $\bar{h}_{g,k}^{(n)}$ (contributed by $\mathbf{x}_{-n}$) and a variable component dependent on $x_n$, i.e., 
\begin{equation}
h_{g,k}(x_n) = \bar{h}_{g,k}^{(n)} + \frac{\sqrt{\eta/N}}{D_{g,k}(x_n)} {\rm e}^{-{\rm j}(k_0 D_{g,k}(x_n)+k_{\rm g}x_n)}.
\end{equation}
The subproblem for optimizing $x_n$ is then formulated as follows:
\begin{equation}\label{prob:1D_search}
x_n^{\star} = \arg \min_{x \in \mathcal{X}_n} \sum_{g=1}^{G} \Big( \max_{k\in\mathcal{K}_g} {\sigma_{g,k}^{2}}{\big|h_{g,k}(x)\big|^{-2}} \Big).
\end{equation}
Since the subproblem \eqref{prob:1D_search} involves only a single scalar variable within a finite discrete set $\mathcal{X}_n$, the global optimum of this subproblem can be efficiently obtained via a one-dimensional grid search. The detailed implementation of the optimization procedure is summarized in Algorithm~\ref{alg:bcd_algorithm}.

The convergence is guaranteed by the monotonicity and boundedness of the objective function. In each iteration, the update of $x_n$ via \eqref{prob:1D_search} ensures a non-increasing objective value, i.e., $f_A(\mathbf{x}^{(t+1)}) \le f_A(\mathbf{x}^{(t)})$. Since $f_A(\mathbf{x})$ is lower-bounded by zero, the iterative sequence  converges to a coordinate-wise local optimum on the discretized set. Regarding complexity, each iteration entails $N$ grid searches of size $L$, where evaluating the objective function requires accounting for all $K$ users. Thus, the total complexity is $\mathcal{O}(I_{\rm iter} N L K)$, which is linear in $N$ and significantly lower than the exponential complexity $\mathcal{O}(K L^N)$ of exhaustive search.

\begin{algorithm}[t]
\caption{Power Allocation and PA Location Optimization for TIN Scheme}
\label{alg:bcd_algorithm}
\begin{algorithmic}[1]
\REQUIRE Power budget $P_{\rm t}$, tolerance $\epsilon$, candidate set $\mathcal{S}_x$, and an initial feasible $\mathbf{x}^{(0)}\in\mathcal{X}$.
\STATE $t\leftarrow 0$, compute $f_{A}(\mathbf{x}^{(0)})$ by \eqref{prob:x_obj}.
\REPEAT
\STATE $t\leftarrow t+1$.
\FOR{$n=1$ to $N$}
\STATE Update $x_n$ by solving \eqref{prob:1D_search}.
\ENDFOR
\STATE Compute $f_{A}(\mathbf{x}^{(t)})$ by \eqref{prob:x_obj}.
\UNTIL{$|f_{A}(\mathbf{x}^{(t)})-f_{A}(\mathbf{x}^{(t-1)})|\le \epsilon$}
\STATE Set $\mathbf{x}^\star\leftarrow \mathbf{x}^{(t)}$.
\STATE Compute $\gamma^{\sf TIN, \star}(\mathbf{x}^\star)$ by \eqref{eq:opt_gamma} and recover $\mathbf{p}^\star$ by \eqref{eq:opt_power}.
\end{algorithmic}
\end{algorithm}

\section{NOMA-Based Transmission}\label{sec:NOMA}
While the TIN scheme is appealing for its low receiver complexity, the MMF performance is severely limited when inter-group interference dominates, since all groups are superposed over the same resource.
To better control such interference without changing the time resource allocation, we adopt power-domain NOMA and exploit SIC, and investigate the resulting joint PA placement and power allocation design.
In what follows, we focus on the NOMA-based multigroup multicast transmission under the MMF objective.

Following the unified MMF formulation in~\eqref{prob:MMF_general}, NOMA occupies the entire frame and thus $\tau_g=1$, $\forall g$. Moreover, the SINR term in \eqref{eq:rate_general} is instantiated by \eqref{eq:SINR_NOMA} under a decoding order $\bm{\pi}$.
Consequently, problem~\eqref{prob:MMF_general} specializes to
\begin{align}\label{prob:MMF_NOMA}
\mathcal{P}^{\sf NOMA}: \max_{\mathbf{x}\in\mathcal{X},\,\mathbf{p}\succeq\mathbf{0}}\ \ &
\min_{g\in\mathcal{G}} \ R_{\pi(g)}^{\sf NOMA}(\mathbf{x},\mathbf{p}),
\quad \text{s.t.} \ \eqref{st:power},
\end{align}
where $R_{\pi(g)}^{\sf NOMA}(\mathbf{x},\mathbf{p})
=\min_{k \in \mathcal{K}_{\pi(g)}} \log_2 \!\left( 1 + \gamma_{\pi(g),k}^{\sf NOMA}(\mathbf{x},\mathbf{p}) \right)$.

Problem~\eqref{prob:MMF_NOMA} is challenging because the decoding permutation $\bm{\pi}$ is coupled with the PA placement $\mathbf{x}$ through the group bottleneck CNRs in \eqref{eq:eff_channel_noma}, which yields a sorting-dependent and hence generally non-smooth objective landscape.
To gain design insights and provide a low-complexity benchmark, we first study the fundamental two-group multicast scenario.

\vspace{-10pt}
\subsection{Two-Group Multicast Scenario}
In the two-group case, the optimal power allocation $\mathbf{p}^{\star}$ can be determined analytically with any PA placements, which allows us to reduce the original joint optimization problem into a problem solely dependent on $\mathbf{x}$.

\subsubsection{Optimal Power Allocation Strategy}
We identify the weak group index $g_{\rm w}$ and the strong group index $g_{\rm s}$ such that their effective CNRs satisfy $A_{g_{\rm s}}(\mathbf{x}) \ge A_{g_{\rm w}}(\mathbf{x})$.

\begin{lemma}\label{prop:G2_closed_form}
For a fixed $\mathbf{x}$, the optimal power allocated to the strong group, denoted by $P^{\star}_{g_{\rm s}}$, can be derived as follows:
\begin{equation}\label{eq:Ps_val}
P^{\star}_{g_{\rm s}} = \frac{\sqrt{(A_{g_{\rm s}} + A_{g_{\rm w}})^2 + 4 P_{\rm t} A_{g_{\rm s}} A_{g_{\rm w}}^2}-(A_{g_{\rm s}} + A_{g_{\rm w}})}{2 A_{g_{\rm s}} A_{g_{\rm w}}}.
\end{equation}
Consequently, the power allocated to the weak group is given by $P^{\star}_{g_{\rm w}} = P_{\rm t} - P^{\star}_{g_{\rm s}}$.
\end{lemma}

\begin{IEEEproof}
Please refer to Appendix~\ref{app:proof_prop_G2} for more details.
\end{IEEEproof}

Substituting the optimal power solution $P^{\star}_{g_{\rm s}}$ from \eqref{eq:Ps_val} back into the SINR expression $\gamma^{{\sf NOMA},\star} = P^{\star}_{g_{\rm s}} A_{g_{\rm s}}$, we obtain the explicit SINR solely in terms of channel gains, which can be written as follows:
\begin{equation}\label{eq:gamma_explicit}
\gamma^{{\sf NOMA},\star} \!=\! \frac{\sqrt{(A_{g_{\rm s}} \!\!+\! A_{g_{\rm w}})^2 \!+\! 4 P_{\rm t} A_{g_{\rm s}} A_{g_{\rm w}}^2}\!\!-\!(A_{g_{\rm s}} \!+\! A_{g_{\rm w}})}{2 A_{g_{\rm w}}}.
\end{equation}
This explicit form reveals a non-linear coupling between the strong and weak groups. 
Specifically, in the high-SNR regime (i.e., $P_{\rm t} \to \infty$), the term $4 P_{\rm t} A_{g_{\rm s}} A_{g_{\rm w}}^2$ dominates the radical, yielding the approximation as follows:
\begin{equation}
\gamma^{{\sf NOMA},\star}(\mathbf{x}) \approx \frac{\sqrt{4 P_{\rm t} A_{g_{\rm s}}(\mathbf{x}) A_{g_{\rm w}}^2(\mathbf{x})}}{2 A_{g_{\rm w}}(\mathbf{x})} = \sqrt{P_{\rm t} A_{g_{\rm s}}(\mathbf{x})}.
\end{equation}

Consequently, the asymptotic multicast rate scales as ${\sf R}^{{\sf NOMA}, \star}(\mathbf{x}) \approx \frac{1}{2} \log_2(P_{\rm t} A_{g_{\rm s}}(\mathbf{x}))$. 
This offers a critical design insight: the multicast rate benefits significantly from the strong group's channel gain $A_{g_{\rm s}}(\mathbf{x})$. 
For the proposed PASS architecture, this characteristic provides a unique optimization flexibility: rather than strictly maximizing the weak group's channel gain $A_{g_{\rm w}}$, the system can alternatively enhance the strong group's channel $A_{g_{\rm s}}$ via PA location reconfiguration to indirectly improve the global fairness rate.

\subsubsection{Element-Wise Sequential PA Location Optimization}
With the optimal power allocation $\mathbf{p}^{\star}(\mathbf{x})$ explicitly given by Lemma~\ref{prop:G2_closed_form}, the original joint optimization problem can be reduced to a problem dependent solely on the PA locations $\mathbf{x}$. 
Specifically, the achievable MMF rate is determined by the balanced SINR, i.e., $\gamma^{{\sf NOMA},\star}(\mathbf{x}) = P_{g_{\rm s}}^{\star}(\mathbf{x}) A_{g_{\rm s}}(\mathbf{x})$. 

For computational efficiency, we adopt the element-wise sequential optimization framework as in Section~III-B. The subproblem for $x_n$ under the NOMA scheme is formulated as follows:
\begin{equation}\label{prob:1D_search_NOMA}
x_n^{\star} = \arg \max_{x \in \mathcal{X}_n} \ P_{g_{\rm s}}^{\star}(\mathbf{x}|_{x_n=x})A_{g_{\rm s}}(\mathbf{x}|_{x_n=x}),
\end{equation}
Since \eqref{prob:1D_search_NOMA} is a univariate optimization problem with an explicit objective function, it can be efficiently solved via one-dimensional grid search.

\vspace{-10pt}
\subsection{General Multicast Scenario}
For the general scenario with $G > 2$ multicast groups, the power allocation problem no longer admits a quadratic closed-form solution. Consequently, the explicit objective function derived for $G=2$ cannot be directly applied. 
To jointly optimize the power allocation and PA locations, we propose an element-wise alternating optimization (AO) framework. Specifically, the PA locations are optimized sequentially, where the optimal power allocation is determined via bisection search within each element-wise update.

\subsubsection{Power Allocation Strategy}
To determine the maximum achievable multicast rate ${\sf R}^{{\sf NOMA}, \star}$, we leverage the monotonic relationship between the data rate and the total transmit power. Specifically, for a target SINR $\gamma$, the minimum required power for group $\pi(k)$ can be derived recursively from the strongest group to the weakest group as follows:
\begin{equation}\label{eq:power_recursive}
P_{\pi(k)}(\gamma) = \gamma \Big( {A^{-1}_{\pi(k)}(\mathbf{x})} + \sum\nolimits_{j=k+1}^G P_{\pi(j)}(\gamma) \Big),
\end{equation}
where the summation term is zero for the strongest group $k=G$. 
The feasibility of $\gamma$ is checked by verifying if the total power constraint~\eqref{st:power} is satisfied.
Since the total required power is strictly increasing with $\gamma$, the optimal multicast rate ${\sf R}^{{\sf NOMA}, \star}(\mathbf{x})$ can be efficiently found via a bisection search over $[\gamma_{\sf min}, \gamma_{\sf max}]$. We initialize the bisection interval as $\gamma_{\sf min}=0$ and $\gamma_{\sf max}=P_{\rm t}\min_{g\in\mathcal{G}} A_{\pi(g)}(\mathbf{x})$, which upper-bounds the optimal equalized NOMA SINR since any group's SINR cannot exceed the SNR obtained when allocating the entire power budget to that group.

\subsubsection{HOE-Enhanced Element-Wise Sequential PA Location Optimization}\label{subsubsec:NOMA_multi_PA}
Given a PA placement vector $\mathbf{x}$, the optimal MMF multicast rate under NOMA can be obtained by the inner procedure, which jointly determines the decoding order $\bm{\pi}(\mathbf{x})$ and the associated power allocation via bisection. The outer-layer design objective is to maximize this optimal rate by adjusting the PA locations, i.e.,
\begin{equation}\label{eq:noma_outer_x}
\max_{\mathbf{x}\in\mathcal{X}} \ {\sf R}^{{\sf NOMA}, \star}(\mathbf{x}).
\end{equation}
Directly evaluating ${\sf R}^{{\sf NOMA}, \star}(\mathbf{x})$ for every candidate position in $\mathcal{X}$ is computationally expensive, since each evaluation involves sorting-dependent SIC ordering and a nested bisection-based power search.
To obtain a low-complexity yet effective PA update rule, we adopt an element-wise sequential optimization strategy and enhance it with a hierarchical objective evaluation (HOE) mechanism, which evaluates the potential of candidate positions in two distinct stages.
Specifically, for the optimization of the $n$th PA with a candidate position $x$, let $R^{\star}_{\sf{best}}$ denote the maximum multicast rate achieved by the optimal solution found so far. The HOE-enhanced element-wise sequential mechanism proceeds as follows:

\begin{itemize}
	\item {Stage I: Upper-Bound Assessment:} 
We first derive a theoretical upper bound for the current candidate, denoted as $R^{\star}_{\sf{upper}}(x)$. This bound corresponds to the ideal rate of the bottleneck group assuming exclusive full power allocation, i.e., 
\begin{equation}\label{eq:upper_bound}
R^{\star}_{\sf{upper}}(x) = \log_2 \Big( 1 + P_{\rm t}\min_{g} A_g(\mathbf{x}|_{x_n=x}) \Big).
\end{equation}
The condition $R^{\star}_{\sf{upper}}(x) \le R^{\star}_{\sf{best}}$ indicates that the candidate position $x$ is inherently suboptimal. In this case, the evaluation terminates at Stage I, and the exact computation is circumvented.
\item {Stage II: Exact Objective Evaluation:} 
Only when $R^{\star}_{\sf{upper}}(x) > R^{\star}_{\sf{best}}$, the evaluation advances to Stage II. The bisection search described in the previous subsection is executed to compute the exact value of ${R}^{\star}(\mathbf{x})$.
\end{itemize}

Mathematically, the update principle for the $n$th PA location under the HOE-enhanced element-wise sequential PA location optimization is expressed as follows:
\begin{align}\label{eq:x_opt_noma}
	x_n^\star=\arg\max_{x\in\mathcal{X}_n}\ \tilde{R}(x),
\end{align}
where
\begin{align}
	\tilde{R}(x)\triangleq
\begin{cases}
R^\star(\mathbf{x}|_{x_n=x}), & R^\star_{\sf upper}(x)>R^\star_{\sf best},\\
-\infty, & \text{otherwise}.
\end{cases}
\end{align}
By prioritizing the low-complexity bound assessment, the computationally intensive Stage II is triggered exclusively for high-potential candidates, which significantly accelerates the convergence of the algorithm. The detailed implementation of this procedure is summarized in Algorithm~\ref{alg:hoe_noma}.

Similar to Algorithm~\ref{alg:bcd_algorithm}, the convergence is guaranteed by the monotonicity of the sequential updates. Regarding complexity, let $\xi \in [0,1]$ denote the candidate retention ratio (i.e., the fraction of candidates selected for Stage II). The total complexity is $\mathcal{O}\big( I_{\rm iter} N L (K + \xi I_{\rm bi} G) \big)$, where $I_{\rm bi}$ is the bisection iterations. Since typically $\xi \ll 1$, the computationally expensive bisection term $\mathcal{O}(I_{\rm bi} G)$ is effectively suppressed for majority of candidates.

\begin{algorithm}[t]
\caption{Joint Power Allocation and PA Location Optimization for NOMA Scheme}
\label{alg:hoe_noma}
\begin{algorithmic}[1]
\REQUIRE Power budget $P_{\rm t}$, tolerance $\epsilon$, discrete grid $\mathcal{S}_x$, and an initial feasible $\mathbf{x}^{(0)}\in\mathcal{X}$.
\STATE Initialize $t\leftarrow 0$. Calculate initial rate $R^{(0)}$.
\REPEAT
\STATE $t\leftarrow t+1$.
\FOR{$n=1$ to $N$}
\STATE For each candidate position $x \in \mathcal{X}_n$: i) Form the temporary PA vector $\tilde{\mathbf{x}}$ by replacing the $n$th element with $x$. ii) Solve the optimal power allocation for $\tilde{\mathbf{x}}$ via bisection search to obtain $R(x)$.
\STATE Update $x_n$ and $\mathbf{p}^{(t)}$ by solving \eqref{eq:x_opt_noma}.
\ENDFOR
\STATE Update current objective value $R^{(t)} = {R}^{\star}(\mathbf{x}^{(t)})$.
\UNTIL{$|R^{(t)}-R^{(t-1)}|\le \epsilon$}
\ENSURE Optimal PA location $\mathbf{x}^\star$ and power allocation $\mathbf{p}^\star$.
\end{algorithmic}
\end{algorithm}

\subsubsection{Single-PA Case}
To gain fundamental insights into the impact of PA placement on the NOMA performance, we consider the special case with a single PA, i.e., $N=1$. We first establish the fundamental relationship between the PA location and the transmit power consumption as follows.

\begin{lemma}\label{prop:closed_form_power}
Consider a PASS-enabled NOMA multigroup multicast system equipped with a single PA located at $x$. To support $G$ groups with a target SINR $\gamma$, the minimum transmit power is given by the closed-form expression as follows:
\begin{equation} \label{eq:power_closed_form}
    P_{\sf{req}}(\gamma, x) = \sum\nolimits_{g=1}^{G} \frac{\gamma (1+\gamma)^{g-1}}{A_{\pi(g)}(x)}.
\end{equation}
\end{lemma}

\begin{IEEEproof}
Please refer to Appendix~\ref{app:power_derivation}.
\end{IEEEproof}

Since the optimization variable $x$ involves only a single scalar bounded within $[0, D_x]$, the globally optimal solution maximizing the multicast rate can be efficiently obtained via a one-dimensional grid search over the discretized aperture $\mathcal{S}_x$.

Based on \eqref{eq:power_closed_form}, we derive two asymptotic design criteria for the optimal PA placement $x^\star$:
\paragraph{Low-SNR Regime ($\gamma \ll 1$)}
Applying the first-order approximation $(1+\gamma)^{g-1} \approx 1$, the power expression simplifies to linear summation $P_{\sf{req}} \approx \gamma f_A(x)$. Consequently, the optimal SINR is approximated by
\begin{equation}
    \gamma^{{\sf NOMA},\star}_{\sf{low}}(x) \approx {P_{\rm t}}f^{-1}_A(x).
\end{equation}
This behavior aligns with the TIN scheme in~\eqref{prob:x_obj}, which focuses on balancing the signal strength of all groups evenly.

\paragraph{High-SNR Regime ($\gamma \gg 1$)}
In this regime, the interference cost dominates, and the term $\gamma(1+\gamma)^{g-1}$ scales as $\gamma^g$. The power requirement is dictated by the most restrictive constraint among all decoding layers. Specifically, for the solution to be feasible, the power consumption of each layer must be within the budget $P_t$. Thus, the achievable SINR is bounded by:
\begin{equation}\label{eq:gamma_noma_high}
    \gamma^{{\sf NOMA},\star}_{\sf{high}}(x) \approx \min\nolimits_{ g \in{\mathcal G}} \left( P_t A_{\pi(g)}(x) \right)^{\frac{1}{g}}.
\end{equation}
This asymptotic bound reveals that NOMA is highly sensitive to \textit{channel disparity} in the high SNR regime. Unlike the low-SNR case, the PA placement here should be designed to ensure that users with higher decoding orders $g$ possess sufficiently strong channel gains $A_{\pi(g)}(x)$ to compensate for the power penalty scaling as $\gamma^g$.

\vspace{-10pt}
\section{TDMA-Based Transmission}\label{sec:TDMA}
To completely eliminate inter-group interference, we employ the TDMA scheme where each group is allocated with a dedicated time slot.
In what follows, we present two PA location optimization protocols, namely PS and PM, which differ in whether the PA locations can be reconfigured across time slots.
\vspace{-10pt}
\subsection{PS-based TDMA Scheme}\label{subsubsec:TDMA_switching}
The TDMA-PS protocol allows slot-wise PA reconfiguration, which assigns a dedicated PA placement $\mathbf{x}_g$ to each group $g$. Consequently, the multicast rate for the $g$th group is
\begin{equation}\label{eq:group_rate_TDMA_switching}
R_{g}^{\sf PS}(\mathbf{x}_g,\!P_g,\!\tau_g)\!=\!\tau_g\! \min_{k\in\mathcal{K}_g}\!\log_2\!\left(1\!+\!{P_g{\sigma_{g,k}^{-2}}\big|h_{g,k}(\mathbf{x}_g)\big|^2}\right).
\end{equation}
The corresponding MMF problem is formulated as follows:
\begin{subequations}\label{prob:MMF_TDMA_switching}
\begin{align}
\max_{\{\mathbf{x}_g\in\mathcal{X}\},\,\mathbf{p},\,\boldsymbol{\tau}}\ &
\min_{g\in\mathcal{G}} \ R_{g}^{\sf PS}(\mathbf{x}_g,P_g,\tau_g), \\
\text{s.t.}\ \
& \sum\nolimits_{g=1}^{G}\tau_g P_g \le P_{\rm t},\ P_g\ge 0,\ \eqref{st:time}. \label{prob:MMF_TDMA_switching_tau}
\end{align}
\end{subequations}
The joint optimization in \eqref{prob:MMF_TDMA_switching} can be efficiently decoupled by exploiting the orthogonality of the PS protocol. Since each group $g$ is served in a dedicated time slot with a group-specific PA placement $\mathbf{x}_g$, the optimization of $\mathbf{x}_g$ depends solely on the local channel conditions and is independent of the resource allocation $\{\tau_g, P_g\}$. This allows the problem to be solved optimally in two sequential stages.
\subsubsection{PA Location Optimization}
Exploiting the decoupling property, the optimal PA placement $\mathbf{x}_g^{\star}$ for each group is obtained by independently maximizing its effective channel gain $A_g(\mathbf{x})$. This effectively decomposes the problem into $G$ single-group multicast communications. Consequently, each $\mathbf{x}_g^{\star}$ can be efficiently determined by employing the element-wise sequential optimization strategy detailed in Section~\ref{subsec:algorithm}.

\subsubsection{Optimal Resource Allocation}
With the optimized PA placement $\mathbf{x}_g^{\star}$ fixed, the original problem~\eqref{prob:MMF_TDMA_switching} reduces to optimizing the time allocation $\boldsymbol{\tau}$ and power allocation $\mathbf{p}$. The objective function involves the coupled term $\tau_g \log_2(1 + P_g A_g(\mathbf{x}_g^{\star}))$, which is generally non-convex with respect to (w.r.t.) $\{\tau_g, P_g\}$.

To tackle this, we introduce the energy variable $E_g \triangleq \tau_g P_g$, which represents the energy consumed by the $g$th group within the TDMA slot. By substituting $P_g = E_g/\tau_g$, the achievable rate of group $g$ is transformed into the following form:
\begin{equation}\label{eq:rate_perspective}
R_g(E_g, \tau_g) = \tau_g \log_2\left( 1 + {\tau_g^{-1}}{E_g A_g(\mathbf{x}_g^{\star})} \right).
\end{equation}
Mathematically, \eqref{eq:rate_perspective} is the perspective function of the concave function $\log_2(1 + x)$, which is jointly concave w.r.t. $E_g$ and $\tau_g$. Consequently, the resource allocation problem can be reformulated as a standard convex optimization problem:
\begin{subequations}\label{prob:convex_resource}
\begin{align}
\max_{t, \,\boldsymbol{\tau}, \,\mathbf{E}}\ \ & t \\[-5pt]
\text{s.t.}\ \ & \tau_g \log_2\left( 1 + {\tau^{-1}_g}{E_g A_g(\mathbf{x}_g^{\star})} \right) \ge t, \quad \forall g \in \mathcal{G}, \label{eq:rate_convex}\\[-3pt]
& \sum\nolimits_{g=1}^{G} E_g \le P_{\rm t}, \ \eqref{prob:MMF_TDMA_switching_tau}. \label{eq:energy_convex}
\end{align}
\end{subequations}
Since problem \eqref{prob:convex_resource} is convex, the globally optimal time allocation $\boldsymbol{\tau}^{\star}$ and energy allocation $\mathbf{E}^{\star}$ can be efficiently obtained using standard convex solvers such as CVX. Finally, the optimal power allocation is recovered by $P_g^{\star} = E_g^{\star}/\tau_g^{\star}$. Regarding complexity, the decoupled scheme entails $\mathcal{O}(N L K)$ for the PA optimization and approximately $\mathcal{O}(G^{3.5})$ for the convex resource allocation via interior-point methods. 

Despite yielding the performance upper bound, the PA switching strategy requires rapid mechanical reconfiguration between time slots, which poses a severe challenge for current PASS hardware. This practicality bottleneck motivates the investigation of the low-complexity PM protocol in the sequel.
\subsection{PM-based TDMA Scheme}\label{subsubsec:TDMA_multiplexing}
The TDMA-PM protocol assigns a common PA placement vector $\mathbf{x}$ to all groups through different slots. Consequently, the multicast rate for the $g$th group reduces to the following:
\begin{equation}\label{eq:group_rate_TDMA_multiplexing}
R_{g}^{\sf PM}(\mathbf{x},\!P_g,\!\tau_g)\!=\!\tau_g\! \min_{k\in\mathcal{K}_g}\!\log_2\!\left(1\!+\!{P_g{\sigma_{g,k}^{-2}}\big|h_{g,k}(\mathbf{x})\big|^2}\right).
\end{equation}
The corresponding MMF problem is formulated as follows:
\begin{align}\label{prob:MMF_TDMA_multiplexing}
\mathcal{P}^{\sf PM}: & \max_{\mathbf{x}\in\mathcal{X},\,\mathbf{p},\,\boldsymbol{\tau}}\ \
\min_{g\in\mathcal{G}} \ R_{g}^{\sf PM}(\mathbf{x},P_g,\tau_g), \
\text{s.t.}\ \eqref{prob:MMF_TDMA_switching_tau}. 
\end{align}
In this protocol, the shared PA locations $\mathbf{x}$ are intricately coupled with the resource allocation variables $\mathbf{p}$ and $\boldsymbol{\tau}$. Unlike the PS protocol where the optimization of $\{{\bf x}_g\}$ and $\{\tau_g, P_g\}$ can be decoupled, any adjustment to a single PA in the PM protocol simultaneously impacts the effective channel gains $\{A_g(\mathbf{x})\}$ of all groups, which renders the joint optimization problem non-convex. Following a similar methodology to Section~\ref{subsubsec:NOMA_multi_PA}, we propose an element-wise sequential optimization for PA locations enhanced by the HOE strategy, while solving the joint power and time allocation within each iteration process.

\subsubsection{Optimal Resource Allocation}
For a fixed PA placement $\mathbf{x}$, the subproblem that maximizes the multicast rate $t$ by jointly optimizing $\mathbf{p}$ and $\boldsymbol{\tau}$ is formulated as follows:
\begin{subequations}\label{prob:inner_convex}
\begin{align}
\max_{t, \,\boldsymbol{\tau}, \,\mathbf{E}}\ \ & t \\[-5pt]
\text{s.t.}\ \ & \tau_g \log_2\left( 1 + {\tau^{-1}_g}{E_g A_g(\mathbf{x})} \right) \ge t, \quad \forall g \in \mathcal{G}, \label{eq:rate_constraint_convex}\\
& \eqref{prob:MMF_TDMA_switching_tau}, \ \eqref{eq:energy_convex}. \notag 
\end{align}
\end{subequations}
Notably, for a target rate $t$, determining its feasibility is equivalent to checking whether the minimum total energy required to support $t$ exceeds the budget $P_{\rm t}$. From \eqref{eq:rate_constraint_convex}, the minimum energy required for group $g$ is derived in a closed form as follows:
\begin{align}\label{eq:min_energy_per_group}
    E_g^{\sf min}(t, \tau_g) = {\tau_g}{A^{-1}_g(\mathbf{x})} ( 2^{t/\tau_g} - 1 ).
\end{align}Consequently, the feasibility check reduces to the following total energy minimization problem w.r.t. $\boldsymbol{\tau}$:
\begin{align}\label{prob:min_total_energy}
    \mathcal{E}(t) = \min_{\boldsymbol{\tau} \succ 0} \sum_{g=1}^{G} \frac{\tau_g}{A_g(\mathbf{x})} \left( 2^{t/\tau_g} - 1 \right), \ \text{s.t.} \ \sum_{g=1}^{G} \tau_g = 1.
\end{align}
Problem \eqref{prob:min_total_energy} is strictly convex. We solve it by introducing the Lagrange multiplier $\nu \ge 0$ for the equality constraint $\sum_g \tau_g = 1$. The Lagrangian function is given by
\begin{align}\label{eq:lagrangian_tdma}
    \mathcal{L}(\boldsymbol{\tau}, \nu) = \sum\nolimits_{g=1}^{G} E_g^{\sf min}(t, \tau_g) + \nu \Big( \sum\nolimits_{g=1}^{G} \tau_g - 1 \Big).
\end{align}
Applying the Karush-Kuhn-Tucker (KKT) conditions, the optimal time allocation $\tau_g^{\star}$ must satisfy the stationarity condition $\frac{\partial \mathcal{L}}{\partial \tau_g} = 0$. By computing the partial derivative of \eqref{eq:min_energy_per_group}, we obtain the optimality equation for each group as follows:
\begin{align}\label{eq:KKT_stationarity}
    \frac{\partial E_g^{\sf min}}{\partial \tau_g} + \nu = {A^{-1}_g(\mathbf{x})} \left( 2^{t/\tau_g} \left( 1 - \frac{t \ln 2}{\tau_g} \right) - 1 \right) + \nu = 0.
\end{align}
Equation \eqref{eq:KKT_stationarity} establishes a monotonic relationship between $\tau_g$ and $\nu$ for a fixed $t$, which implies that $\tau_g(\nu)$ is uniquely determined by $\nu$. Since the total time consumption $\sum_g \tau_g(\nu)$ is also monotonic w.r.t. $\nu$, the optimal dual variable $\nu^\star$ satisfying $\sum \tau_g(\nu^\star)=1$ can be efficiently found via a bisection search.
Finally, the maximum achievable rate $t^\star$ is obtained by an outer-loop bisection search, which iteratively adjusts $t$ and checks the energy feasibility condition $\mathcal{E}(t) \le P_{\rm t}$.

\subsubsection{HOE-Enhanced Element-Wise Sequential PA Location Optimization}\label{pare:hoe_pm}
We adopt the same HOE-enhanced sequential optimization framework established in Section~\ref{subsubsec:NOMA_multi_PA} to optimize $\mathbf{x}$. The update rule follows the structure of \eqref{eq:x_opt_noma}, with the exact objective evaluation of Stage II replaced by the KKT-based inner solver derived above.
Moreover, for the Stage I, we derive the theoretical limit by assuming the bottleneck group exclusively occupies the entire time resource and power budget, i.e., $\tau_g=1$, $P_g=P_{\rm t}$, which leads to a mathematical expression identical to \eqref{eq:upper_bound}.

Similar to Algorithm~\ref{alg:bcd_algorithm}, the convergence is guaranteed by the monotonicity of the sequential updates. 
Specifically, for any candidate placement $\tilde{\mathbf{x}}$, the inner resource-allocation solver returns the globally optimal value by solving \eqref{prob:inner_convex}. 
Therefore, each element-wise update that selects the best $x\in\mathcal{X}_n$ yields a non-decreasing outer objective sequence, i.e., $R^\star(\mathbf{x}^{(t+1)}) \ge R^\star(\mathbf{x}^{(t)})$. 
Since $R^\star(\mathbf{x})$ is upper-bounded under the finite resource budgets, the overall algorithm converges.
Regarding complexity, the total complexity is $\mathcal{O}\!\left(I_{\rm iter} N L \big(K + \xi I_{\rm out} I_{\nu} G\big)\right)$,
where $I_{\rm out}$ is the number of outer bisection iterations over the target rate $t$, and $I_{\nu}$ is the number of inner bisection iterations for solving \eqref{eq:KKT_stationarity}.

\subsubsection{Single-PA Case}
To gain more design insights and facilitate a low-complexity implementation, we focus on the single-PA scenario with equal time allocation, i.e., $\tau_g = 1/G, \forall g$. In this case, the effective bottleneck CNR $A_g(\mathbf{x})$ defined in \eqref{eq:eff_channel_noma} becomes a scalar function of $x$.
The original problem \eqref{prob:MMF_TDMA_multiplexing} can thus be simplified to a joint optimization problem of power allocation and PA location. The following lemma provides the optimal power allocation structure in closed form.

\begin{lemma}\label{prop:TDMA_closed_form}
For the considered TDMA-PM scheme with $N=1$ and equal time allocation, given any fixed PA location $x$, the optimal power allocation $P_g^{\star}(x)$ that maximizes the multicast rate is given by
\begin{equation}\label{eq:TDMA_optimal_power}
P_g^{\star}(x) = \frac{ G P_{\rm t} }{{A_g(x)}f_A({x}) } .
\end{equation}
Substituting \eqref{eq:TDMA_optimal_power} back into the objective function yields the maximized multicast rate as follows:
\begin{equation}\label{eq:TDMA_optimal_rate}
R^{\sf PM, \star}(x) = {G}^{-1} \log_2 ( 1 + {G P_{\rm t}}{f^{-1}_A({x})} ).
\end{equation}Consequently, the optimal PA location $x^{\star}$ is obtained by solving the equivalent minimization problem as follows:
\begin{equation}\label{eq:TDMA_x_opt_problem}
x^{\star} = \arg \min_{x \in \mathcal{X}} \sum\nolimits_{g=1}^{G} \left( \max_{k\in\mathcal{K}_g} {\sigma_{g,k}^{2}}{\big|h_{g,k}(x)\big|^{-2}} \right).
\end{equation}
\end{lemma}

\begin{IEEEproof}
Please refer to Appendix~\ref{app:proof_prop_TDMA}.
\end{IEEEproof}
\subsection{Further Discussion}
Despite distinct MA schemes, TIN, TDMA-PM, and NOMA when operating in the low-SNR regime share a common PA optimization principle: the minimization of the sum of inverse effective group CNR $f_A({x})$. This finding establishes a fundamental guideline that PASS-enabled multigroup multicast systems should prioritize maximizing the {\textit{harmonic mean}} of group effective gains to ensure fairness.
However, the effectiveness of this unified geometric metric diverges fundamentally due to the intrinsic limitations of each scheme:
\begin{itemize}
    \item {TIN Scheme (Interference-Limited):} Although minimizing $f_A({x})$ balances the signal strengths, the performance is fundamentally bounded by the interference derived in \eqref{eq:tin_ceiling_rate}. In the high SNR regime, the gain from optimizing PA location becomes negligible, as the inter-group interference dominates.
    \item {TDMA Scheme (Resource-Limited):} This scheme circumvents interference saturation, which allows the geometric gain to translate directly into SNR improvement. However, the achievable rate is strictly penalized by the multiplexing factor $1/G$, which makes it less spectrally efficient than TIN or NOMA in the high-SNR regime.
    \item {NOMA Scheme (SNR-Adaptive Behavior):} The unification holds in the low-SNR regime because the system is noise-limited, which leads to a linear power-sharing behavior similar to TIN. Conversely, in the interference-limited high-SNR regime, the optimal PA placement must adhere to the alternative geometric criterion derived in~\eqref{eq:gamma_noma_high}. This design prioritizes strong channels to support high-order SIC decoding, which leverages channel disparity into a performance gain.
\end{itemize}

\begin{algorithm}[t]
\caption{Joint Resource Allocation and PA Location Optimization for TDMA-PM Scheme}
\label{alg:hoe_tdma_pm}
\begin{algorithmic}[1]
\REQUIRE Power budget $P_{\rm t}$, tolerance $\epsilon$, discrete grid $\mathcal{S}_x$, and an initial feasible $\mathbf{x}^{(0)}\in\mathcal{X}$.
\STATE Initialize $t\leftarrow 0$. Calculate initial rate $R^{(0)}$.
\REPEAT
\STATE $t\leftarrow t+1$.
\FOR{$n=1$ to $N$}
\STATE \quad i) Form the temporary PA vector $\tilde{\mathbf{x}}$ by replacing the $n$th element with $x$. ii) Solve the optimal time and power allocation for $\tilde{\mathbf{x}}$ via the KKT-based inner solver to obtain $R(x)$.
\STATE Update $x_n \leftarrow \arg \max_{x \in \mathcal{X}_n} R(x)$.
\STATE Update the optimal resource allocation $\mathbf{p}^{(t)}$ and $\boldsymbol{\tau}^{(t)}$.
\ENDFOR
\STATE Update current objective value $R^{(t)} = {R}^{\sf PM, \star}(\mathbf{x}^{(t)})$.
\UNTIL{$|R^{(t)}-R^{(t-1)}|\le \epsilon$}
\ENSURE Optimal PA location $\mathbf{x}^\star$, power allocation $\mathbf{p}^\star$, and time allocation $\boldsymbol{\tau}^\star$.
\end{algorithmic}
\end{algorithm}

\vspace{-10pt}
\section{Numerical Results}\label{sec:simulation}
To validate the effectiveness of the proposed optimization algorithms for joint PA placement and scheme-specific resource allocation, extensive simulations were conducted across the TIN, NOMA, and TDMA schemes.
\vspace{-10pt}
\subsection{Simulation Setup}
For the configuration of PASS, the waveguide is deployed at a height of $h = 5$ m with the side length $D_{\rm y} = 6$ m. The PAs are deployed along the waveguide with a minimum spacing of $\Delta_{\sf min} = 0.5\lambda$. The effective refractive index of the dielectric waveguide is assumed to be $n_{\mathrm{eff}} = 1.44$~\cite{10945421}.

For the multicast configuration, the carrier frequency is set to $f_c = 28$ GHz. Unless otherwise specified, the users are assumed to be uniformly distributed within the service region $D_{\rm x} \times D_{\rm y}$. The noise power at each user is set to $\sigma^2 = -90$ dBm. All numerical results are obtained by averaging over $1000$ independent random channel realizations.

For the algorithm configuration, the $N$ PAs are initialized with random positions along the waveguide, which ensures that the minimum inter-element spacing constraint $\Delta_{\sf min}$ is satisfied. To facilitate the one-dimensional search for the PA optimization, the waveguide is discretized into $L = 200$ equally spaced candidate points. The stopping tolerance for the proposed element-wise sequential optimization framework and the inner KKT-based solver is set to $\epsilon = 10^{-4}$. The maximum number of iterations for the outer loop is restricted to $20$ to ensure low computational latency.
\vspace{-10pt}
\subsection{Baseline Scheme}
We consider a conventional fixed-array system as a baseline, where the $N$ PAs are arranged as a uniform linear array (ULA) centered at the waveguide aperture, with a fixed inter-element spacing of $\lambda/2$. 
All antennas are connected to a single RF chain through a phase shifter to implement analog beamforming. Specifically, the continuous phase shift of each antenna is quantized into $L$ discrete candidate values uniformly distributed over $[0, 2\pi)$.
Based on this discretization, we employ an element-wise sequential phase optimization strategy. 
For any given analog beamforming, the optimal power and time resource allocations are determined using the algorithms as proposed in PASS.
\vspace{-10pt}
\subsection{Convergence Behavior of Proposed Algorithms}
\begin{figure}[!t]
\centering
\includegraphics[height=0.26\textwidth]{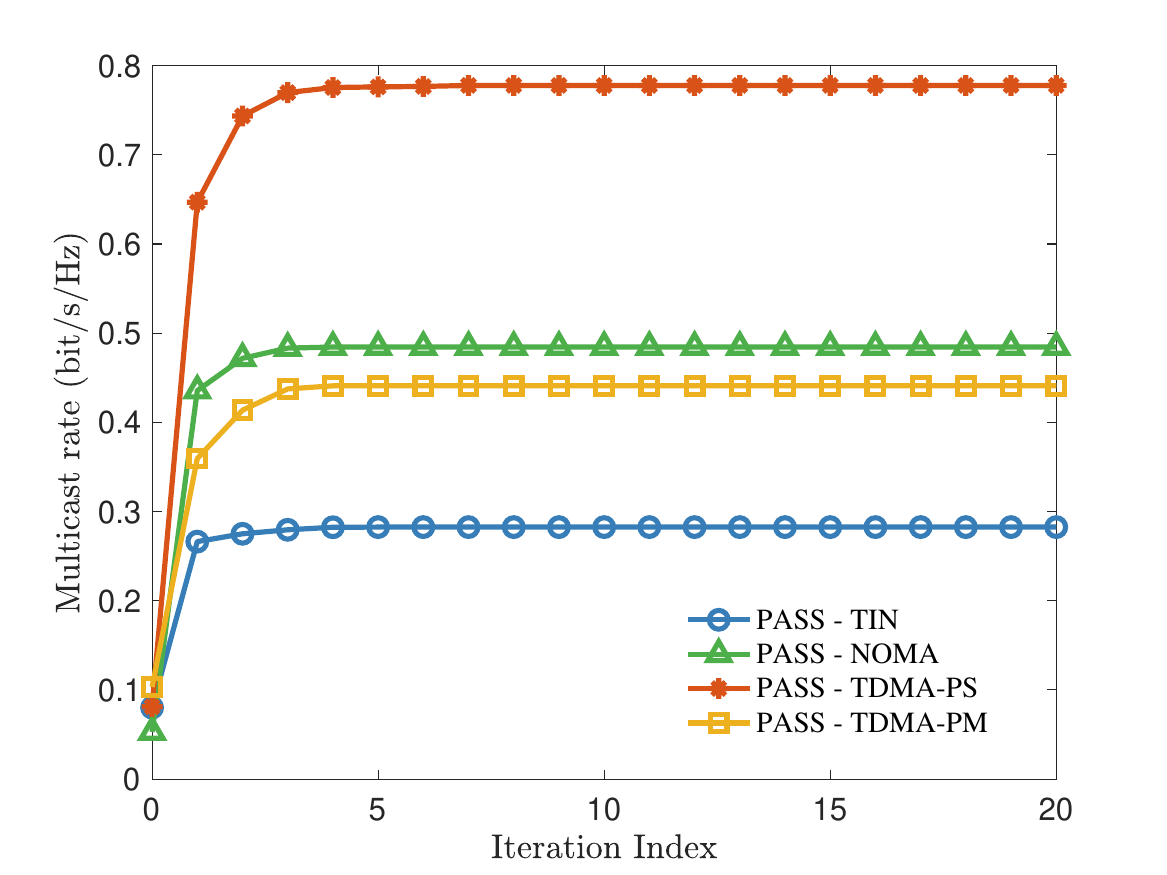}
\caption{Convergence behavior of the proposed optimization algorithms in terms of the multicast rate with $N=10$, $G=4$, $K=12$, $D_{\rm x}=20$ m, and $P_{\rm t}=-10$~dBm.}
\label{Fig:Convergence}
\vspace{-5pt}
\end{figure}
Fig. \ref{Fig:Convergence} illustrates the convergence behavior of the proposed algorithms in terms of the multicast rate. First, at the initial stage with random PA placement, all schemes exhibit comparable and limited multicast performance, which highlights the necessity of PA location optimization. Second, the multicast rates of all schemes stabilize within a small number of iterations, which validates the high efficiency of the adopted element-wise sequential optimization strategy. This implies that the proposed algorithms are well-suited for real-time wireless communication scenarios. Third, the strictly monotonic increasing trend verifies the validity of the AO framework, which ensures that each update step leads to a non-decreasing objective function value.
\vspace{-10pt}
\subsection{Multicast Rate versus the Transmit Power}
\begin{figure}[!t]
\centering
\includegraphics[height=0.26\textwidth]{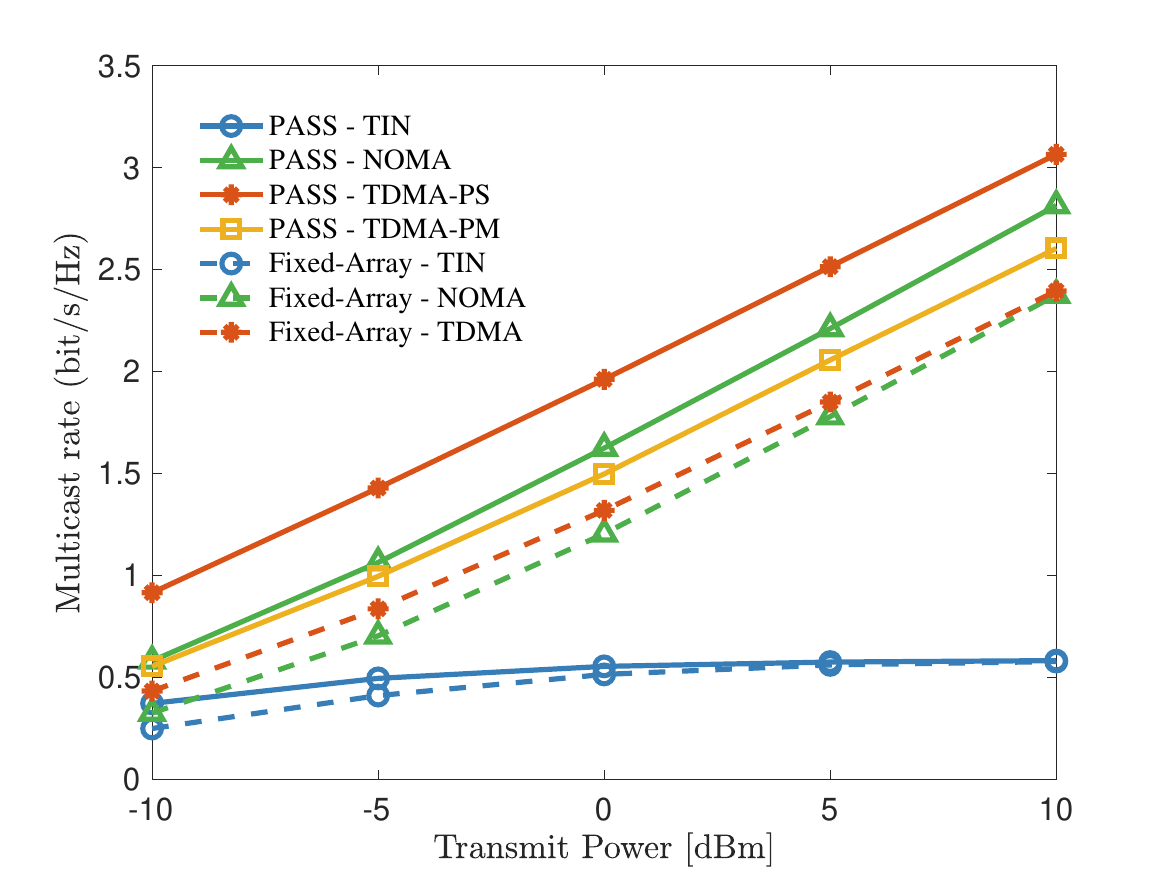}
\caption{Multicast rate versus the transmit power under the random group distribution with $N=10$, $G=4$, $K=12$, $D_{\rm x}=20$ m.}
\label{Fig:Power_0}
\vspace{-5pt}
\end{figure}

\begin{figure}[!t]
\centering
\includegraphics[height=0.26\textwidth]{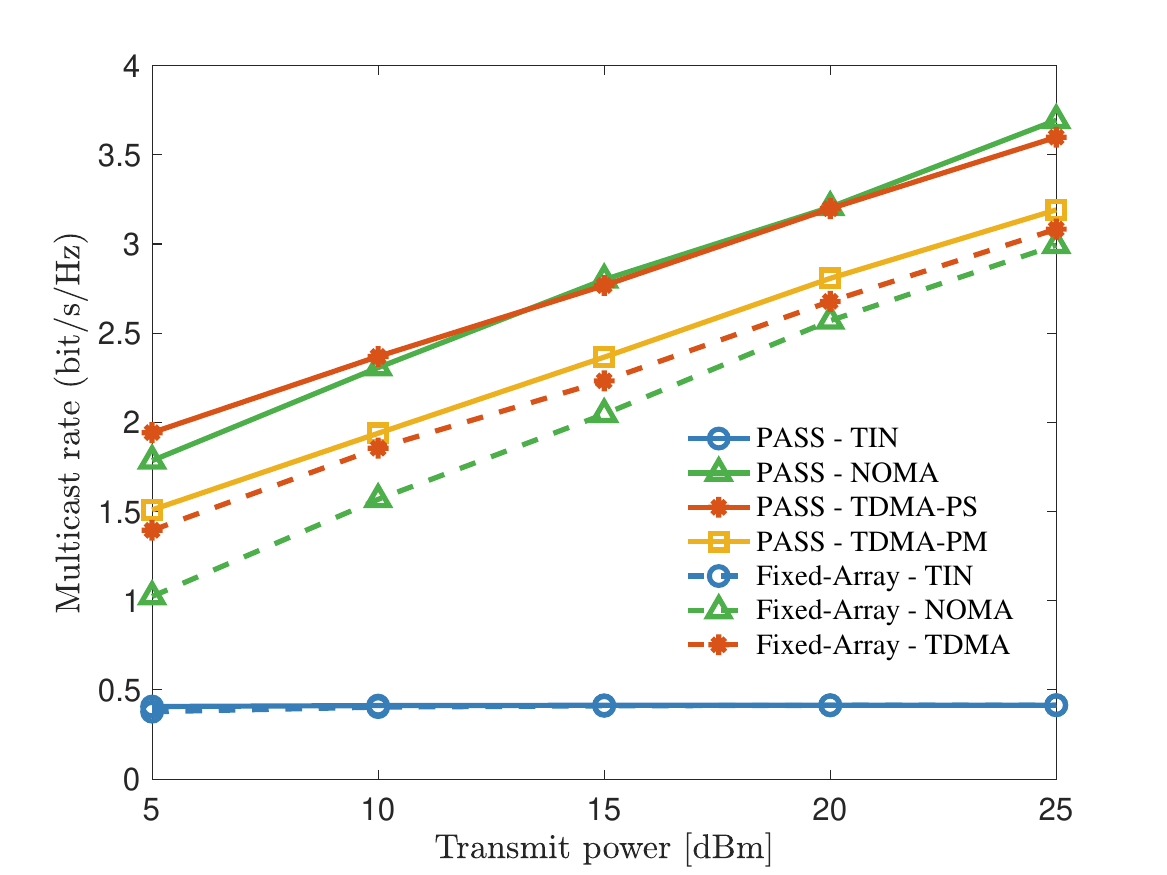}
\caption{Multicast rate versus the transmit power under the heterogeneous group distribution with $N=10$, $G=4$, $K=12$, $D_{\rm x}=20$ m.}
\label{Fig:Power}
\vspace{-5pt}
\end{figure}
Fig.~\eqref{Fig:Power_0} investigates the multicast rate versus the transmit power. In this scenario, users from all groups are randomly distributed within an overlapping service region. By optimizing the PA locations, the PASS architecture can effectively balance the channel gains, which implies that the bottleneck users in each group experience comparable channel quality. The results indicate that TDMA-PS achieves the superiority performance, which is attitude to its orthogonal access nature and the fully exploiting of flexible reconfigurable of PASS. Although TDMA-PM also eliminates inter-group interference, it falls behind NOMA due to its time resource limited characteristic. Furthermore, the performance of NOMA gradually approaches that of TDMA-PS as the transmit power increases, which demonstrates its effectiveness in high-SNR regimes even without PA switching.

To further evaluate the distinct advantages across different schemes, Fig.~\ref{Fig:Power} presents the system performance under a heterogeneous group distribution, where the users are clustered into spatially disjoint regions. In this setup, the channel conditions vary significantly across groups, and NOMA outperforms TDMA-PS in the high-SNR regime. This superiority is attributed to NOMA's ability to exploit channel disparity.

Finally, it is worth noting that for the TIN scheme, both PASS and the fixed-array system share a common performance upper bound, which is given in~\eqref{eq:tin_ceiling_rate}. This is because, any optimization of PA locations to enhance the useful signal strength inevitably leads to a concurrent increase in the inter-group interference. Consequently, the spatial advantage of PASS over conventional fixed-array systems naturally vanishes.
\vspace{-10pt}
\subsection{Multicast Rate versus the Side Length $D_{\rm x}$}
\begin{figure}[!t]
\centering
\includegraphics[height=0.26\textwidth]{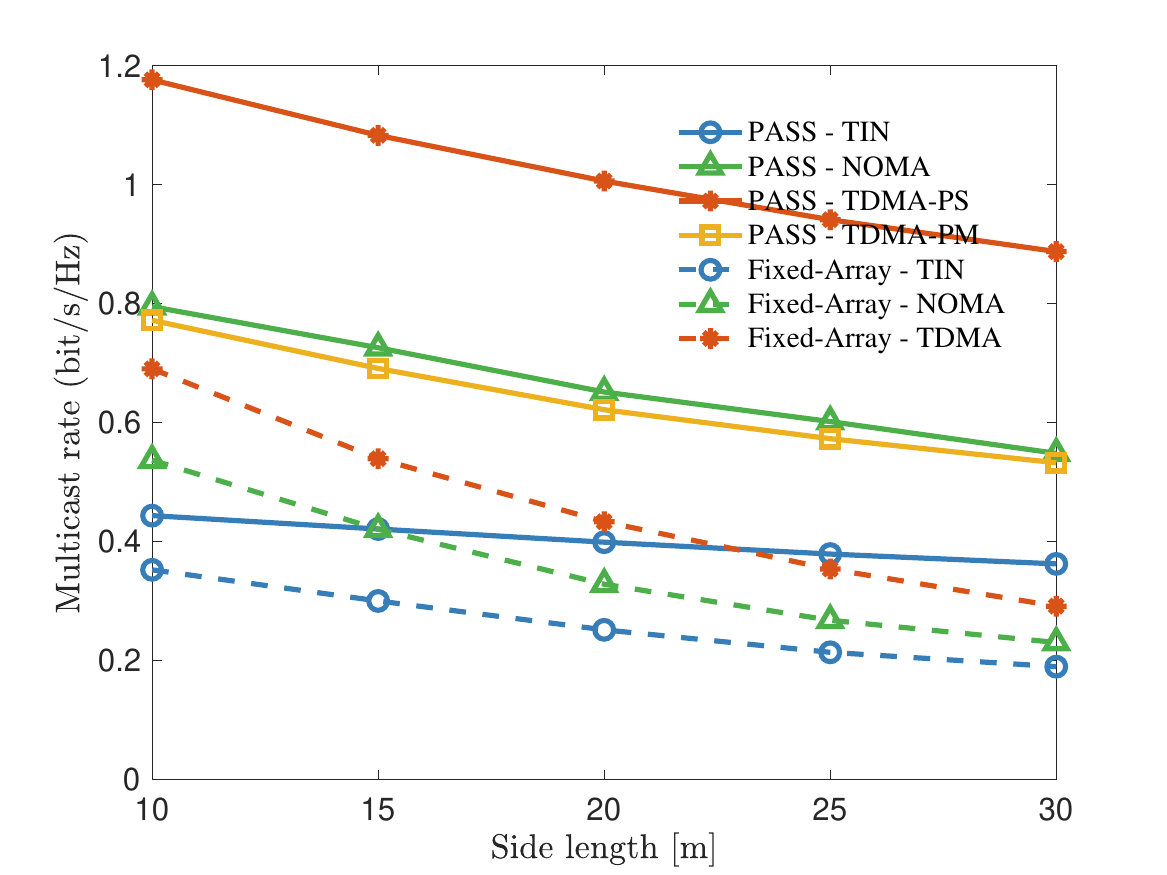}
\caption{Multicast rate versus the side length $D_{\rm x}$ with $N=10$, $G=3$, $K=12$, and $P_{\rm t}=-10$~dBm.}
\label{Fig:Dx}
\vspace{-5pt}
\end{figure}
Fig.~\ref{Fig:Dx} investigates the impact of the service region size $D_{\rm x}$ on the system performance. Two key observations can be drawn from the results.
First, the proposed PASS architecture consistently outperforms the conventional fixed-array benchmark across all multiple access schemes. This superiority demonstrates the exceptional capability of PASS in combating large-scale fading. By flexibly adjusting the PAs along the waveguide, PASS can effectively shorten the communication distances, thereby establishing strong LoS links for distributed users. Notably, the performance gap between PASS and the fixed-array becomes more pronounced as $D_{\rm x}$ increases. This is because, for the fixed-array system, expanding the service region inevitably increases the average distance between the array and the users, which leads to severe path loss degradation. In contrast, the geometric reconfigurability of PASS allows the antennas to track and move closer to the users, which maintains a relatively short average distance even in larger service areas. This highlights the robustness of PASS in larger service regions.

Second, it is observed that the multicast rate of PASS also exhibits a slight decreasing trend with the increase of $D_{\rm x}$. This arises from the tradeoff between spatial coverage and array gain. As the user distribution becomes more sparse over a larger $D_{\rm x}$, the PAs are forced to disperse more widely to satisfy the MMF criterion among distinct groups. Consequently, the ability of the PAs to concentrate energy in a specific user direction is weakened, and the pinching array gain for each user is diminished, which results in a slight degradation in the achievable rate.
\vspace{-5pt}
\subsection{Multicast Rate versus the Number of PAs}
\begin{figure}[!t]
\centering
\includegraphics[height=0.26\textwidth]{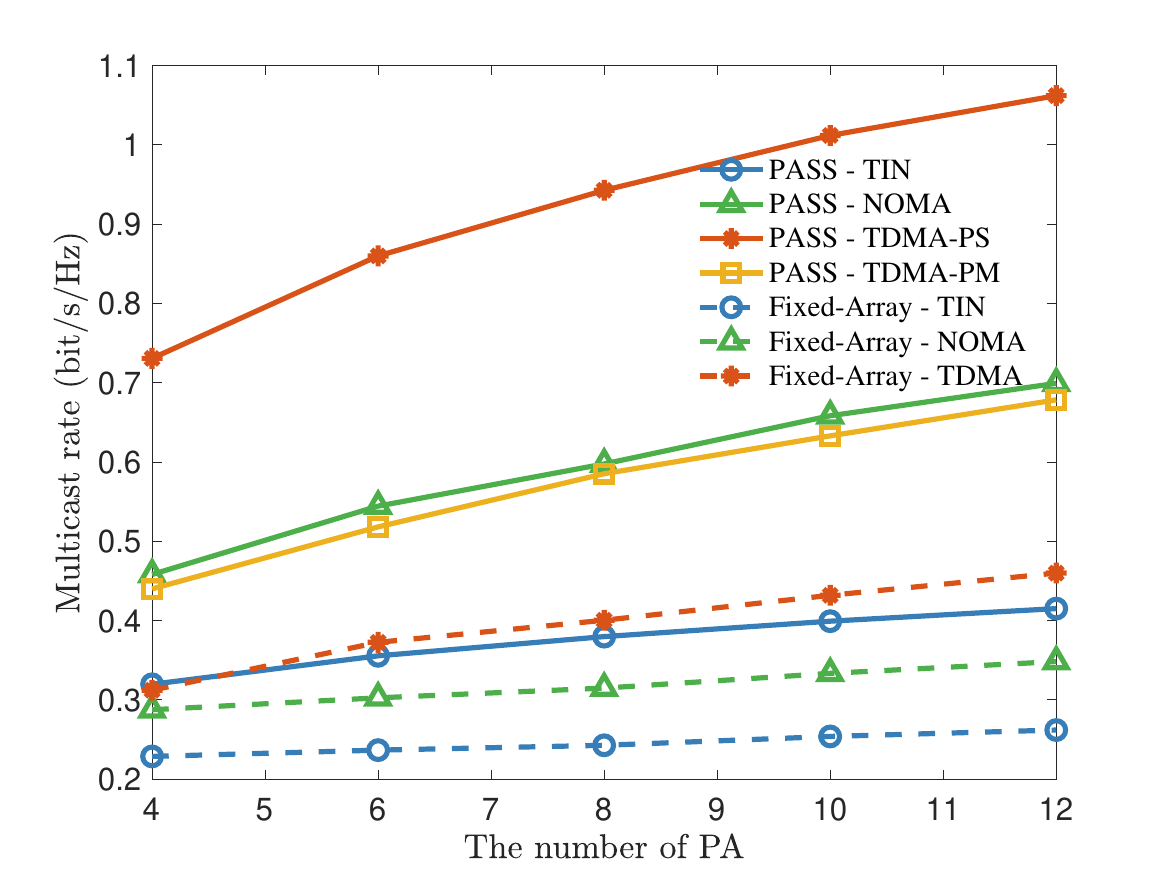}
\caption{Multicast rate versus the number of PAs with $D_{\rm x} = 20$ m, $G=3$, $K=12$, and $P_{\rm t}=-10$~dBm.}
\label{Fig:PA}
\vspace{-5pt}
\end{figure}
Fig.~\ref{Fig:PA} illustrates the impact of the number of PAs on the multicast rate. Two key insights can be derived from the results. First, a monotonic increase in achievable rate is observed across all schemes as $N$ increases. This trend verifies that, under the MMF formulation, a larger number of PAs allow PASS to more effectively balance the channel conditions among all users. Specifically, the PASS architecture gains stronger capabilities to reconstruct the array geometry with more PAs, which establishes robust LoS links for the users.

Second, the performance gap between PASS and the fixed-array system enlarges significantly as $N$ increases. This demonstrates the superior capability of PASS in combating large-scale fading. Unlike fixed arrays that rely solely on beamforming gain, PASS leverages the increased number of movable PAs to flexibly shorten the average communication distance. Consequently, the advantage of PASS becomes more significant with a larger number of PAs.

\subsection{Multicast Rate versus the Number of Multicast Group}
\begin{figure}[!t]
\centering
\includegraphics[height=0.26\textwidth]{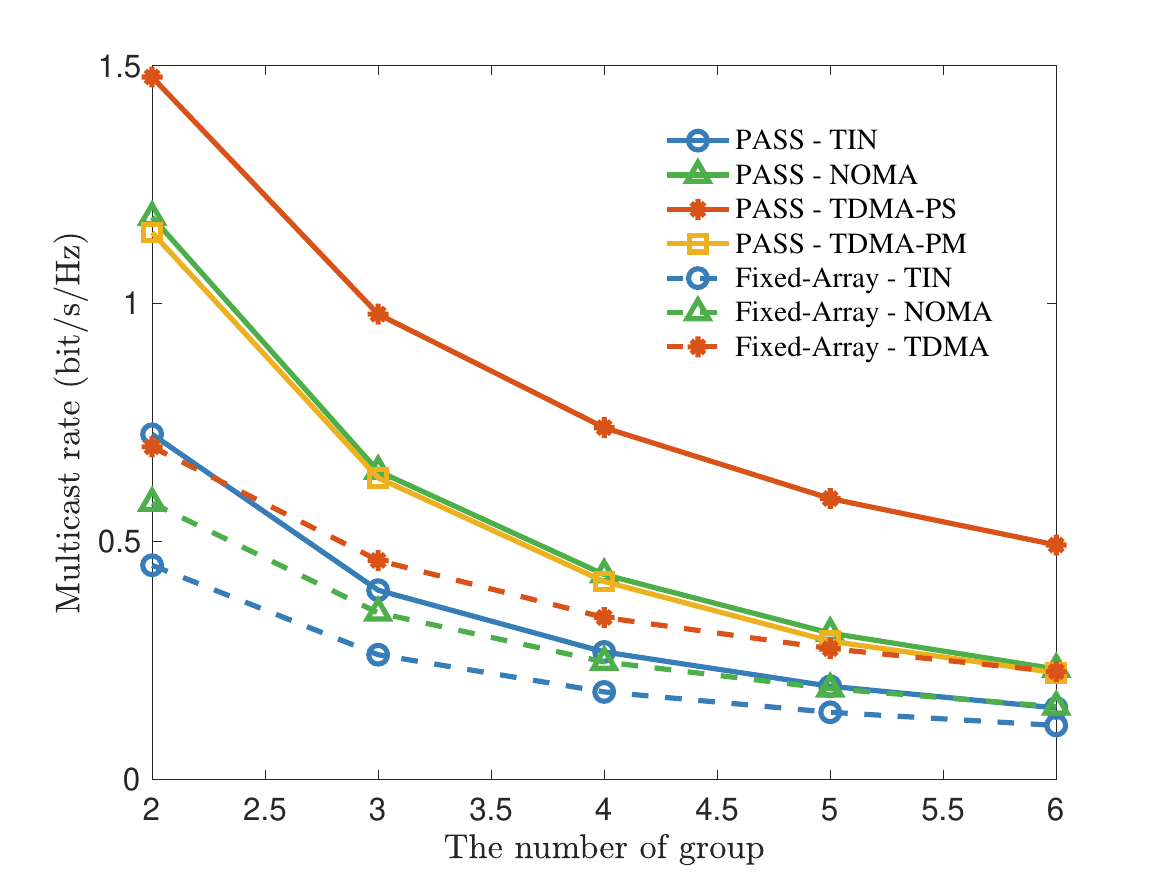}
\caption{Multicast rate versus the number of multicast groups with $D_{\rm x} = 20$~m, $N = 10$, $K=4$ in each group, and $P_{\rm t}=-10$~dBm.}
\label{Fig:Group}
\vspace{-5pt}
\end{figure}
Fig.~\ref{Fig:Group} depicts the impact of the number of user groups on the multicast rate. As anticipated, the system performance for all schemes exhibits a monotonic decline as $G$ increases. For the TDMA schemes, the performance deterioration is primarily due to the reduction in time resources allocated to each group. Regarding the TIN scheme, the performance decline arises from the intensification of inter-group interference. Finally, the performance degradation for NOMA is caused by spatial energy dispersion. Since NOMA serves all groups simultaneously, increasing $G$ forces the PAs to distribute signal energy across more directions to ensure fairness. This attenuates the  pinching array gain allocated to specific users, which inevitably degrades the signal strength for each group.

\section{Conclusion}\label{sec:conclusion}
This paper investigated the max-min fairness problem in PASS-enabled multigroup multicast systems under three MA schemes: TIN, NOMA, and TDMA. For the TIN scheme, we derived a closed-form optimal power allocation together with an analytical upper bound on the achievable rate, and showed that the MMF-optimal PA placement is obtained by maximizing the harmonic mean of the channel gains. For NOMA and TDMA, we proposed HOE-enhanced sequential optimization frameworks that jointly design PA locations and resource allocation with low computational complexity. Numerical results verified the effectiveness of the proposed algorithms and revealed that PASS offers clear gains over fixed-array baselines. Among the three schemes,TDMA-PS delivers the strongest performance by eliminating inter-group interference and fully exploiting group-wise spatial reconfigurability. NOMA consistently outperforms TDMA-PM and tends to be preferable to TDMA-PS in high-power heterogeneous-group scenarios, where its interference management capability yields pronounced fairness gains. In contrast, TIN acts as a simple lower-bound benchmark that remains attractive for systems constrained by receiver complexity or hardware cost.

\appendices
\section{Proof of Lemma~\ref{prop:closed_form_power_TIN}} \label{app:proof_prop1}

With fixed PA locations $\mathbf{x}$, the effective group channel gains $\{A_g(\mathbf{x})\}_{g=1}^G$ are constant. For notational brevity, we omit the dependency on $\mathbf{x}$ in this proof. The MMF power allocation problem can be equivalently reformulated by introducing a slack variable $t$ representing the minimum SINR as follows:
\begin{align}\label{prob:app_mmf}
\max_{t, \,\mathbf{p}}\ t, \quad\text{s.t.}\ \  \frac{P_g A_g}{(P_{\rm t} - P_g) A_g + 1} \ge t,  \forall g\in\mathcal{G}, \ \eqref{st:power}.
\end{align}
The Lagrangian function of~\eqref{prob:app_mmf} is formulated as follows:
\begin{align}
\mathcal{L}(t, \mathbf{p}, \boldsymbol{\lambda}, \mu) \!=\!& -t \!+\!\! \sum\nolimits_{g=1}^{G} \!\!\lambda_g \Big( t (P_{\rm t} A_g + 1) \!-\! P_g A_g (1+t) \Big) \notag\\
&+ \mu \Big( \sum\nolimits_{g=1}^{G} P_g - P_{\rm t} \Big),
\end{align}
where $\boldsymbol{\lambda} = [\lambda_1, \dots, \lambda_G]^{\sf T}$ with $\lambda_g>0, \forall g\in{\mathcal G}$, and $\mu \ge 0$ are the Lagrange multipliers associated with the SINR constraints and the total power constraint, respectively. Since the original problem is a convex optimization problem, the Karush-Kuhn-Tucker (KKT) conditions are necessary and sufficient for optimality. The KKT conditions are given by
\begin{subequations}
\begin{align}
&\frac{\partial \mathcal{L}}{\partial t} = -1 + \sum\nolimits_{g=1}^{G} \lambda_g \Big( P_{\rm t} A_g + 1 - P_g A_g \Big) = 0, \label{eq:kkt_t} \\
&\frac{\partial \mathcal{L}}{\partial P_g} = -\lambda_g A_g (1+t) + \mu = 0, \quad \forall g\in{\mathcal G}, \label{eq:kkt_p} \\
&\mu \Big( \sum\nolimits_{g=1}^{G} P_g - P_{\rm t} \Big) = 0, \label{eq:kkt_slack_mu}\\
&\lambda_g \Big( t (P_{\rm t} A_g + 1) - P_g A_g (1+t) \Big) = 0, \quad \forall g\in{\mathcal G}. \label{eq:kkt_slack_lambda}
\end{align}
\end{subequations}

From \eqref{eq:kkt_p}, we have $\lambda_g = \frac{\mu}{A_g(1+t)}$. Since $A_g > 0$ and $t > 0$, a strictly positive power budget $P_{\rm t} > 0$ implies that the power constraint must be active at the optimum, i.e., $\mu > 0$. Consequently, we have $\lambda_g > 0$ for all $g \in \mathcal{G}$. 
The strict positivity of $\lambda_g$ invokes the complementary slackness condition in \eqref{eq:kkt_slack_lambda}, which forces the SINR constraints to be active for all groups. Thus, all groups achieve the equalized SINR $t$, which can be given by
\begin{equation}
P_g = \frac{t}{1+t} \left( P_{\rm t} + {A^{-1}_g} \right), \quad \forall g\in{\mathcal G}.
\end{equation}
Substituting the optimal $P_g$ into the active total power constraint~\eqref{st:power} yields
\begin{equation}\label{eq:t_before}
\sum\nolimits_{g=1}^{G} \Big( \frac{t}{1+t} \left( P_{\rm t} + {A^{-1}_g} \right) \Big) = P_{\rm t}.
\end{equation}
Finally, we obtain
\begin{align}
	t = \Big(\sum\nolimits_{g=1}^{G} \left( 1 +{P^{-1}_{\rm t} A^{-1}_g} \right) - 1\Big)^{-1}.
\end{align}
Solving for $t$ yields the optimal SINR $\gamma^{\sf TIN, \star}(\mathbf{x})$ in \eqref{eq:opt_gamma}, and substituting $t$ back into the expression for $P_g$ yields \eqref{eq:opt_power}. This completes the proof.

\section{Proof of Lemma~\ref{prop:G2_closed_form}} \label{app:proof_prop_G2}

According to the MMF criterion, the optimal solution satisfies two conditions: 1) the total power constraint is active, i.e., $P_s + P_w = P_{\rm t}$; and 2) the strong and weak groups achieve identical SINRs.
Equating the SINR of the strong group and the SNR of weak group yields
\begin{equation}\label{eq:sinr_eq}
P_s A_s = \frac{P_w A_w}{P_s A_w + 1}.
\end{equation}
Substituting $P_w = P_{\rm t} - P_s$ into \eqref{eq:sinr_eq}, we obtain
\begin{equation}
P_s A_s (P_s A_w + 1) = (P_{\rm t} - P_s) A_w.
\end{equation}It follows that $(A_s A_w) P_s^2 + (A_s + A_w) P_s - P_{\rm t} A_w = 0$. Since $P_s\geq 0$, the feasible solution can be given by
\begin{equation}
P_s^{\star} = \frac{-(A_s + A_w) + \sqrt{(A_s + A_w)^2 + 4 P_{\rm t} A_s A_w^2}}{2 A_s A_w}.
\end{equation}
Finally, the weak group power is determined by $P_w^{\star} = P_{\rm t} - P_s^{\star}$. This completes the proof.
\vspace{-10pt}
\section{Proof of Lemma~\ref{prop:closed_form_power}}
\label{app:power_derivation}
Let the multicast groups be indexed by the decoding order, where group $g$ decodes before group $j$ for $g < j$. To achieve a target SINR $\gamma$, the transmit power $P_g$ for group $g$ must satisfy the following SINR condition:
\begin{align}\label{eq:app_sinr_constraints}
\frac{P_g A_g}{1 + A_g \sum_{j=g+1}^{G} P_j} \ge \gamma, \quad \forall g\in{\mathcal G},
\end{align}
where the interference term is zero for the last user $g=G$.

To minimize the total power, the inequality in \eqref{eq:app_sinr_constraints} must hold with equality. Let $S_g \triangleq \sum_{j=g}^{G} P_j$ denote the partial sum of power from group $g$ to $G$. Then, we have $P_g = S_g - S_{g+1}$ (with $S_{G+1}=0$) and $\sum_{j=g+1}^{G} P_j = S_{g+1}$.
Substituting these into \eqref{eq:app_sinr_constraints}, the condition can be rewritten as follows:
\begin{align}
\frac{(S_g - S_{g+1}) A_g}{1 + A_g S_{g+1}} = \gamma.
\end{align}
Rearranging the terms to establish a recursive relationship for $S_g$, which can be written as follows:
\begin{align}
(S_g - S_{g+1}) A_g &= \gamma (1 + A_g S_{g+1}) \notag \\
S_g &= \frac{\gamma}{A_g} + (1+\gamma) S_{g+1}. \label{eq:app_recursion}
\end{align}
Equation \eqref{eq:app_recursion} represents a linear recurrence relation. We verify the solution by backward induction starting from $g=G$:
\begin{itemize}
    \item For $g=G$, since $S_{G+1}=0$, we have $S_G = \frac{\gamma}{A_G}$. This matches the term for $k=0$ in the summation if we view the expansion from the end.
    \item For $g=G-1$, substituting $S_G$:
\begin{equation}
        S_{G-1} = \frac{\gamma}{A_{G-1}} + (1+\gamma)\frac{\gamma}{A_G}.
    \end{equation}
\end{itemize}
By recursively expanding \eqref{eq:app_recursion} from $g=1$, the total required power $P_{\rm req} = S_1$ is obtained as follows:
\begin{align}
P_{\sf req} &= \frac{\gamma}{A_1} + (1+\gamma)\left( \frac{\gamma}{A_2} + (1+\gamma)S_3 \right) \notag \\
&= \frac{\gamma}{A_1} + \frac{\gamma(1+\gamma)}{A_2} + \dots + \frac{\gamma(1+\gamma)^{G-1}}{A_G} \notag \\
&= \sum\nolimits_{g=1}^{G} \frac{\gamma (1+\gamma)^{g-1}}{A_g}.
\end{align}
This completes the proof.

\section{Proof of Lemma~\ref{prop:TDMA_closed_form}}\label{app:proof_prop_TDMA}

With a fixed PA location $x$ and equal time allocation $\tau_g=1/G$, the original MMF problem can be recast into the standard epigraph form. By introducing an auxiliary variable $t$ representing the equalized minimal rate, the problem is formulated as follows:
\begin{subequations}\label{prob:appendix_primal}
\begin{align}
\max_{ \{P_g\}, \, t } \quad & t \\
\text{s.t.} \quad & {G^{-1}}\log_2\left( 1 + P_g A_g(x) \right) \ge t, \quad \forall g\in{\mathcal G}, \label{eq:app_rate_constraint}\\
& \sum\nolimits_{g=1}^{G} {G^{-1}} P_g \le P_{\rm t}, \label{eq:app_power_constraint}
\end{align}
\end{subequations}
where \eqref{eq:app_power_constraint} corresponds to the average power constraint. 
The rate constraint \eqref{eq:app_rate_constraint} can be rearranged as $P_g \ge {A^{-1}_g(x)} (2^{Gt} - 1)$. 
The Lagrangian of problem \eqref{prob:appendix_primal} is given by
\begin{align}
\mathcal{L} \!=\! -t + \!\!\sum\nolimits_{g=1}^{G} \!\mu_g\!\! \left( \frac{2^{Gt} - 1}{A_g(x)} - P_g \right) \!+\! \lambda \!\left( \sum\nolimits_{g=1}^{G}\!\! P_g \!-\! G P_{\rm t}\! \right),
\end{align}
where $\mu_g \ge 0$ and $\lambda \ge 0$ are the Lagrange multipliers associated with the rate constraints and the total power constraint, respectively. 
The KKT conditions w.r.t. $P_g$ and $t$ are derived as follows:
\begin{subequations}
\begin{align}
& \frac{\partial \mathcal{L}}{\partial t} = -1 + G (\ln 2) 2^{Gt} \sum\nolimits_{g=1}^{G} \frac{\mu_g}{A_g(x)} = 0, \\
& \frac{\partial \mathcal{L}}{\partial P_g} = -\mu_g + \lambda = 0, \quad \forall g\in{\mathcal G}, \label{eq:deriv-Pg} \\
& \lambda \Big( \sum\nolimits_{g=1}^{G} P_g - G P_{\rm t} \Big) = 0.
\end{align}
\end{subequations}
From~\eqref{eq:deriv-Pg}, we have $\mu_g = \lambda$ for all $g$. Since the optimal MMF solution requires all rate constraints to be active, the terms inside the square brackets in the Lagrangian must be zero. This yields the relationship as follows:
\begin{equation}
P_g = {A^{-1}_g(x)} \left(2^{Gt} - 1\right).
\end{equation}
Let $C = 2^{Gt} - 1$. The above equation indicates that the optimal power $P_g$ is strictly proportional to the inverse of the effective channel gain $A^{-1}_g(x)$.
Substituting this into the active power constraint $\sum P_g = G P_{\rm t}$ and defining $f_A({x}) \triangleq \sum\nolimits_{g=1}^{G} {A^{-1}_g({x})}$, we obtain
\begin{equation}
C f_A({x}) = G P_{\rm t} \implies C = {G P_{\rm t}}{f^{-1}_A({x})}.
\end{equation}
Substituting $C$ back into the expression for $P_g$, the closed-form power allocation is given by 
\begin{equation}
P_g^{\star}(x) = {{A^{-1}_g(x)}}f^{-1}_A({x}) G P_{\rm t}.
\end{equation}
Finally, substituting $P_g^{\star}(x)$ back into the objective function, we obtain the maximized MMF rate as follows:
\begin{equation}
t^{\star} = {G^{-1}} \log_2 \left( 1 + {G P_{\rm t}}f^{-1}_A({x}) \right).
\end{equation}
Since $\log(\cdot)$ is a monotonically increasing function, maximizing $t^{\star}$ is equivalent to minimizing the denominator term $f_A({x})$. Thus, the optimal PA location $x^{\star}$ is determined by minimizing the sum of inverse effective CNRs, which completes the proof.

\bibliographystyle{IEEEtran} 
\bibliography{reference}    

\end{document}